%% file: quantum_final.tex
\documentclass[a4paper,onecolumn,11pt,accepted=2025-04-16]{quantumarticle}
\pdfoutput=1
\usepackage[utf8]{inputenc}
\usepackage[english]{babel}
\usepackage[T1]{fontenc}
\usepackage{mathtools}
\usepackage{amssymb}
\usepackage{amsmath}
\usepackage{amsthm}
\usepackage{physics}
\usepackage{tikz}
\usepackage{lipsum}
\usepackage[numbers]{natbib}
\newcommand{\PSPACE}{\mathsf{PSPACE}}
\newcommand{\QMA}{\mathsf{QMA}}
\newcommand{\QCMA}{\mathsf{QCMA}}

\newcommand{\BQP}{\mathsf{BQP}}

\newcommand{\PP}{\mathsf{PP}}

\newcommand{\SBQP}{\mathsf{SBQP}}
\newcommand{\NEXP}{\mathsf{NEXP}}
\newcommand{\poly}{\mathsf{poly}}
\newcommand{\ALL}{\mathsf{ALL}}

\newcommand{\tnorm}[1]{\left\| #1 \right\|_{\mathrm{tr}}}
\newcommand{\opnorm}[1]{\left\| #1 \right\|_{\mathrm{op}}}
\newcommand{\dnorm}[1]{\left\| #1 \right\|_{\diamond}}

\newcommand{\mO}{\mathcal{O}}

\newcommand{\conv}{\textup{conv}}
\newcommand{\eig}{\textup{eig}}

\usepackage{thmtools}
\usepackage{thm-restate}

\declaretheorem[name=Lemma]{lemma}
\declaretheorem[name=Corollary]{corollary}
\declaretheorem[name=Definition]{definition}
\declaretheorem[name=Proposition]{proposition}
\declaretheorem[name=Claim]{claim}

\declaretheorem[name=Remark]{remark}

\newtheorem{subdefinition}{Definition}[definition]

\usepackage[pagebackref]{hyperref}
\usepackage[nameinlink,capitalize]{cleveref}

\crefname{claim}{claim}{claims}
\crefname{claim}{Claim}{Claims}

\begin{document}

\title{Lower Bounds for Unitary Property Testing with Proofs and Advice}

\author{Jordi Weggemans}
\affiliation{QuSoft \& CWI, Amsterdam, the Netherlands}
\orcid{0000-0002-8469-6900}
\email{jrw@cwi.nl}
\homepage{https://jordiweggemans.github.io}

\maketitle

\begin{abstract}
\noindent In unitary property testing a quantum algorithm, also known as a \textit{tester}, is given query access to a black-box unitary and has to decide whether it satisfies some property. We propose a new technique for proving lower bounds on the quantum query complexity of unitary property testing and related problems, which utilises its connection to unitary channel discrimination. The main advantage of this technique is that all obtained lower bounds hold for any $\mathsf{C}$-tester with $\mathsf{C} \subseteq \QMA(2)\slash\mathsf{qpoly}$, showing that even having access to \textit{both} (unentangled) quantum proofs and quantum advice does not help for many unitary property testing problems. 

We apply our technique to prove lower bounds for problems like quantum phase estimation, the entanglement entropy problem, quantum Gibbs sampling and more, removing all logarithmic factors in the lower bounds obtained by the sample-to-query lifting theorem of Wang and Zhang (2023). As a direct corollary, we show that there exist quantum oracles relative to which $\QMA(2) \not\supset \SBQP$ and $\QMA\slash\mathsf{qpoly} \not\supset \SBQP$. The former shows that, at least in a black-box way, having unentangled quantum proofs does not help in solving problems that require high precision.
\end{abstract}

\section{Introduction}
\label{sec:int}
Quantum query complexity is the study of how many queries a quantum algorithm has to make to some black-box input $X$ to decide whether $X$ satisfies some property $\mathcal{P}$. Whilst quantum query complexity conventionally focuses on $X$ being inherently classical (i.e.,~a classical bit string), it is also possible to consider the setting where $X$ is some black-box unitary. Whilst the former (which we call classical property testing) is very useful in obtaining insights into the differences in computational power between different classes of computation, classical or quantum, the latter (called unitary property testing) gives another way to compare inherently quantum classes. These problems, first studied by Wang~\cite{wang2011property}, got considerably more attention recently~\cite{she2022unitary,chen2023testing,wang2023quantum}.

Query complexities can vastly differ among different classes of computational models. For example, the search problem, which is to decide whether a string of length $N$ is either the all-zeros string or has at least one entry with a `$1$', is known to have classical query complexity of $\Theta(N)$ and quantum query complexity of $\Theta(\sqrt{N})$~\cite{grover1996fast,bennett1997strengths}. However, given a string by an untrusted prover, the query complexity of the search problem is just $1$ in both cases. A similar result holds for a unitary property testing analogue of search as introduced by Aaronson and Kuperberg~\cite{aaronson2007quantum}, where one has to decide whether a given black-box unitary $U$ applies either the identity operation $\mathbb{I}$ or the reflection $\mathbb{I}-2\ketbra{\psi}$ for some unknown {$N$-dimensional} quantum state $\ket{\psi}$. This problem has in general a quantum query complexity of $\Theta(\sqrt{N})$, but can again be solved by just a single query if a \emph{quantum state} is provided by an untrusted prover as an extra input. For many other unitary property testing problems, it is unclear whether quantum proofs and/or trusted advice states might help in solving these tasks. 

Our main contribution is a new lower bound technique for query complexity of unitary property testing in the presence of \textit{both} (quantum) proofs and advice, based on the connection between unitary property testing and the discrimination of unitary quantum channels, a topic widely studied in quantum information theory~\cite{acin2001statistical,d2001using,duan2007entanglement,ziman2010single}.   

\subsection{Unitary property testing lower bounds by unitary channel discrimination}
Following the definition of She and Yuen~\cite{she2022unitary}, a unitary \emph{property testing} problem $\mathcal{P}$ consists of two disjoint sets of $d$-dimensional unitaries $\mathcal{P}_\textup{yes}$ and $ \mathcal{P}_\textup{no}$ for a fixed dimension $d$. For some unknown input unitary $U$, given the promise that either (i) $U \in \mathcal{P}_\textup{yes}$ or (ii) $ U \in \mathcal{P}_\textup{no}$, the task is to decide whether (i) or (ii) holds by making (controlled) queries to $U$ or its inverse. Similarly, in unitary channel \emph{discrimination} one is also given black-box access to an unknown unitary $U$ promised to be from a known set of candidate unitaries, also known as hypotheses, from which one has to decide which one corresponds to $U$. When the hypothesis set consists of two families of unitaries $\{U_1: U_1 \in \mathcal{P}_\textup{yes}\}$ and $\{U_2: U_2 \in \mathcal{P}_\textup{no}\}$ for some property $\mathcal{P}$, discriminating between any $U_1 \in \mathcal{P}_\textup{yes}$ and $U_2 \in \mathcal{P}_\textup{no}$ allows one to solve the property testing problem.

We consider quantum query algorithms for property testing (also known as $\mathsf{C}$-\textit{testers}, where $\mathsf{C}$ refers to a complexity class) according to the following model: starting with an initial state $\ket{\psi_\textup{init}}$, the tester is allowed to make queries to $U$ interleaved with applications of unitaries $V_t$ that do not depend on the input, and finally makes a two-outcome measurement to decide with high probability whether $U = U_1$ or $U = U_2$. Again, the queries to $U$ can be controlled, applications of the inverse, or a combination of both.  The types of initial states we consider are of the form $ \ket{\psi_\textup{init}} = \ket{0}^{\otimes \poly(n)} \ket{\psi_\text{input}}$, where $\ket{\psi_\text{input}}$ depends on the class of property testers considered. For example, a $\BQP$-tester has no input state, a $\BQP\slash\mathsf{qpoly}$-tester has $\ket{\psi_\text{input}}=\ket{\psi_n}$, where $\ket{\psi_n}$ is some trusted quantum state which only depends on the input length (defined below), and a $\QMA$-tester has $\ket{\psi_\text{input}}=\ket{\xi}$ for some untrusted quantum proof $\ket{\xi}$. The strongest advice/proof model we consider is then a $\QMA(2)\slash\mathsf{qpoly}$-tester, which uses both a polynomial number of unentangled quantum proofs and a quantum advice state as extra input to the computation.

It is well-known that the one-shot distinguishability (i.e.~when one can only apply the quantum channel once) is characterised by the so-called diamond distance between $U_1$ and $U_2$, denoted as $\frac{1}{2}\dnorm{\mathcal{U}(U_1) - \mathcal{U}(U_2)}$.\footnote{This notation is used to make a distinction between the unitary $U$ and its corresponding channel $\mathcal{U}(U) : \rho \rightarrow U^\dagger \rho U$, where $\rho$ is a density matrix.} The diamond norm can be computed via a semidefinite program~\cite{watrous_2018}, and allows for simplified expressions when the channels are unitaries~\cite{ziman2010single}. Having controlled access to a unitary can make two indistinguishable channels perfectly distinguishable. Take for example $U_1 = \mathbb{I}$ and $U_2 = -\mathbb{I}$: we have that $\frac{1}{2}\dnorm{\mathcal{U}(U_1) - \mathcal{U}(U_2)} = 0$, but $\frac{1}{2}\dnorm{\mathcal{U}(cU_1) - \mathcal{U}(cU_2)} = 1$, where $cU_{i} = \ket{0}\bra{0} \otimes \mathbb{I} + \ket{1}\bra{1} \otimes U_{i}$ for $i \in \{1,2\}$. A circuit which achieves perfect discrimination is one which applies $cU_{i}$ to the Bell state $\ket{\Phi^{+}}$ and performs a measurement in the Bell basis. However, in this simple example, the power of the controlled application of $U$ can easily be removed by applying a global phase to one of the unitaries, such that now $U_1 = U_2 = \mathbb{I}$, making them indistinguishable even in the controlled setting. As it turns out, adding this variable global phase to one of the unitaries (which is fine in our lower bound method, since we have the freedom to choose our $U_1$ and $U_2$) allows one to show that access to the inverse, controlled access, or a combination of both, does not increase the ability to discriminate both unitaries. 

Using the above ideas, our main theorem is a lower bound on the query complexity of any quantum algorithm, given access to multiple unentangled quantum proofs and quantum advice, that discriminates between the unitaries $U_1$ and $U_2$ with success probability $\geq 2/3$ given a promise on the diamond distance between $U_1$ and $U_2$.

\begin{restatable}[Diamond-norm lower bound for unitary channel discrimination]{theorem}{lowb} Let $\theta \in [0, 2\pi)$ and let $U_1,U_2 \in \mathbb{U}(d)$ such that $0 \notin \conv(\eig(U_1^\dagger U_2))$. Now let $U \in \{U_1,U_2^\theta\}$ with $U_2^\theta = e^{i\theta} U_2$ be a unitary to which one has black-box access, including controlled operations, applications of the inverse and a combination of both. Suppose one has to decide whether (i) $U = U_1$ or (ii) $U = U_2^{\theta}$ holds, promised that either one of them is the case and $\frac{1}{2}\dnorm{\mathcal{U}(U_1)-\mathcal{U}(U_2)} \leq \epsilon$. Then there exists a $\theta \in [0,2 \pi)$ such that to decide with success probability $\geq 2/3$ whether (i) or (ii) holds, any $\mathsf{C}$-tester with $\mathsf{C} \subseteq \QMA(2)\slash\mathsf{qpoly}$  needs to make at least
\begin{align*}
   T \geq \Omega\left(\frac{1}{\epsilon}\right)
\end{align*}
queries to $U$.
\label{thm:lb_technique}
\end{restatable}

\noindent Essentially,~\cref{thm:lb_technique} establishes in a unitary query setting what is known as \emph{Heisenberg-limited scaling} (or the \emph{Heisenberg limit}) in quantum sensing and metrology~\cite{holland1993interferometric,ou1996complementarity}. This implies that achieving a $1/N$-factor improvement in accuracy requires $\sim N$ additional ``resources,'' which, in our case, correspond to queries. Since the Heisenberg limit is an information-theoretic notion, it should apply universally---including computational settings that allow (quantum) proofs and advice.

Variations of~\cref{thm:lb_technique} were already proven in other works~\cite{acin2001statistical,duan2007entanglement}, but to our knowledge not in the setting where (i) one considers access to quantum proofs and/or advice and (ii) a computational model which allows for access to the inverse, controlled and controlled inverse of $U$. Moreover, we would like to stress that the key contribution of this work is to give a unified exposition of the idea\footnote{Proving lower bounds by reducing the discrimination and identification tasks is a standard tool in theoretical computer science.} that~\cref{thm:lb_technique} can be used as a general framework to prove lower bounds for unitary property testing problems in the setting where one is given advice as well as proofs, by adopting the following strategy:
\begin{enumerate}
    \item Given a unitary property $\mathcal{P} = (\mathcal{P}_\textup{yes},\mathcal{P}_\textup{no})$, find two unitaries $U_1$, $U_2$ such that $U_1 \in \mathcal{P}_\textup{yes}$ and $U_2 \in \mathcal{P}_\textup{no}$. 
    \item Show that a unitary property tester for $\mathcal{P}$ implies a distinguisher for $U_1$ and $U_2$.
    \item Prove a lower bound on the channel discrimination complexity of $U_1,U_2$ using~\cref{thm:lb_technique}, which implies a lower bound on the query complexity of the unitary property tester.
\end{enumerate}

\noindent As we will show in the main body of the paper, this strategy turns out to be a powerful procedure to obtain simple, yet often tight lower bounds for a wide range of unitary property testing and related problems. For example, the lower bound of quantum phase estimation as proven by Bessen~\cite{bessen2005lower}, relies on frequency analysis and takes up a 7-page double-column paper. Using our technique, an optimal lower bound in a stronger setting (including proofs and advice) can be shown using a proof of just 7 lines(!).\footnote{Of course, this is a slight cheat in the way we count, but hopefully this captures the gist that this technique---once it is set up---is powerful in its simplicity.}

\tikzset{every picture/.style={line width=0.75pt}} 

\begin{figure}
    \centering

\tikzset{every picture/.style={line width=0.75pt}} 

\begin{tikzpicture}[x=0.75pt,y=0.75pt,yscale=-1,xscale=1]

\draw    (405.33,20.5) -- (574,20.5) ;
\draw    (405.33,30.5) -- (574,30.5) ;
\draw    (405.33,40.5) -- (574,40.5) ;
\draw    (405.33,60) -- (574,60) ;
\draw    (405.33,71) -- (574,71) ;
\draw    (405,81) -- (573.67,81) ;
\draw    (405,91) -- (573.67,91) ;
\draw    (405,111) -- (573.67,111) ;
\draw    (119,21.5) -- (371,21.5) ;
\draw    (119,31.5) -- (371,31.5) ;
\draw    (119,41.5) -- (371,41.5) ;
\draw    (119,61) -- (371,61) ;
\draw    (119,72) -- (371,72) ;
\draw    (118.5,82) -- (370.5,82) ;
\draw    (118.5,92) -- (370.5,92) ;
\draw    (118.5,112) -- (370.5,112) ;
\draw  [fill={rgb, 255:red, 255; green, 255; blue, 255 }  ,fill opacity=1 ] (129,12) -- (198.89,12) -- (198.89,122.75) -- (129,122.75) -- cycle ;
\draw  [fill={rgb, 255:red, 255; green, 255; blue, 255 }  ,fill opacity=1 ] (209,11.6) -- (278.89,11.6) -- (278.89,98.25) -- (209,98.25) -- cycle ;
\draw  [fill={rgb, 255:red, 255; green, 255; blue, 255 }  ,fill opacity=1 ] (290.5,11.6) -- (360.39,11.6) -- (360.39,123.75) -- (290.5,123.75) -- cycle ;
\draw  [fill={rgb, 255:red, 255; green, 255; blue, 255 }  ,fill opacity=1 ] (415,12.4) -- (484.89,12.4) -- (484.89,97.25) -- (415,97.25) -- cycle ;
\draw  [fill={rgb, 255:red, 255; green, 255; blue, 255 }  ,fill opacity=1 ] (494.5,12) -- (564.39,12) -- (564.39,122.25) -- (494.5,122.25) -- cycle ;
\draw  [fill={rgb, 255:red, 255; green, 255; blue, 255 }  ,fill opacity=1 ] (574.75,11.49) -- (594.27,11.49) -- (594.27,27.15) -- (574.75,27.15) -- cycle ;
\draw    (577.02,22.92) .. controls (579.29,19.7) and (581.25,16.2) .. (583.85,16.2) .. controls (586.46,16.2) and (589.06,19.7) .. (591.34,23.5) ;
\draw  [fill={rgb, 255:red, 0; green, 0; blue, 0 }  ,fill opacity=1 ] (591.88,14.12) -- (590.62,17.71) -- (588.34,15.5) -- cycle ;
\draw    (583.28,22.73) -- (591.88,14.12) ;

\draw  [dash pattern={on 0.84pt off 2.51pt}]  (124.33,46) -- (124.33,57.17) ;
\draw  [dash pattern={on 0.84pt off 2.51pt}]  (123.33,96) -- (123.33,107.17) ;

\draw (150.62,47.81) node [anchor=north west][inner sep=0.75pt]    {$V^{0}$};
\draw (373.48,51.47) node [anchor=north west][inner sep=0.75pt]    {$\dotsc $};
\draw (231.12,41.81) node [anchor=north west][inner sep=0.75pt]    {$\widetilde{U^{1}} \ \ $};
\draw (315.62,47.31) node [anchor=north west][inner sep=0.75pt]    {$V^{1}$};
\draw (437.12,40.31) node [anchor=north west][inner sep=0.75pt]    {$\widetilde{U^{T}} \ \ $};
\draw (516.12,48.81) node [anchor=north west][inner sep=0.75pt]    {$V^{T}$};
\draw (34.33,85.23) node [anchor=north west][inner sep=0.75pt]    {$\ket{\psi _{\text{input}}}$};
\draw (31.67,30.4) node [anchor=north west][inner sep=0.75pt]    {$\ket{0}^{\otimes \poly( n)}$};
\draw (97.33,16.83) node [anchor=north west][inner sep=0.75pt]   [align=left] {{\Huge {\fontfamily{pcr}\selectfont \}}}};
\draw (97.5,71.83) node [anchor=north west][inner sep=0.75pt]   [align=left] {{\Huge {\fontfamily{pcr}\selectfont \}}}};

\end{tikzpicture}
    \caption{Query complexity model for unitary property tester making $T$ queries to a unitary of interest $U$. The initial state is $\ket{\psi_\textup{init}} = \ket{0}^{\otimes \poly(n)} \ket{\psi_\text{input}}$ is allowed to consist of an input-independent part and an input-dependent part, depending on the class $\mathsf{C}$ of the $\mathsf{C}$-property tester. The unitary of interest $U$ can be accessed directly, through its inverse, controlled or controlled inverse, i.e.~we have that $\tilde{U}^t \in \{U_i,U_i^\dagger, cU_i, cU_i^\dagger \}$ for all $t \in [T]$, $i \in \{1,2\}$.}
    \label{fig:query-model}
\end{figure}
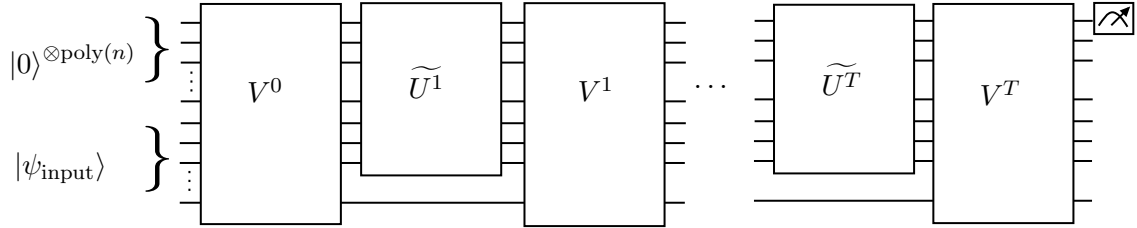

\subsection{Applications}
We consider many of the same unitary property testing problems as in the work by Wang and Zhang~\cite{wang2023quantum}, except for the \textit{subset support verification} problem, which is a new problem we introduce in this work. Let us briefly sketch the definitions of all four unitary property testing problems considered:
\begin{itemize}
    \item \textbf{Quantum phase estimation:} given access to a unitary $U$ and an eigenstate $\ket{\psi}$, determine whether the eigenphase $\theta : U \ket{\psi} = e^{2 i \pi \theta } \ket{\psi}$ satisfies $\theta \geq b$ or $\leq a$ for $b-a = \epsilon$.
    \item \textbf{Entanglement entropy:} given access to a bipartite state $\ket{\psi}_{AB} \in \mathbb{C}^d \otimes \mathbb{C}^d $ through a unitary $U = \mathbb{I} - \ketbra{\psi}$, determine whether $S_2 (\tr_B(\ketbra{\psi}) ) \leq a$  or $\geq b$ with $b-a = \Delta$, where $
    S_2(\rho) = - \ln[\tr[\rho^2]]$.
    \item \textbf{Subset support verification:} given access to a subset $S \subseteq \{0,1\}^n$ through a unitary $U : U \ket{0^n} = \frac{1}{\sqrt{|S|}}\sum_{i \in S} \ket{i}$, determine whether $j \in S$ or $j \notin S$ for some string $j \in \{0,1\}^n$.
    \item \textbf{Quantum amplitude estimation:}  given a unitary $U$ which acts as $U \ket{0^n} = \sqrt{\alpha} \ket{\psi} +\sqrt{1-\alpha} |\psi^\perp\rangle$, determine whether $\sqrt{\alpha} \geq b$ or $\leq a$ with $b-a = \epsilon$.
\end{itemize}
Additionally, we also consider the following four unitary problems, which are not unitary property testing problems:
\begin{itemize}
    \item \textbf{Thermal state preparation:} given access to some Hamiltonian $H$ through a block-encoding $U_H$ prepare (an approximation of) the Gibbs state $e^{- \beta H}/\tr[e^{- \beta H}]$ at inverse temperature $\beta$.
    \item \textbf{Hamiltonian simulation:}  given access to some Hamiltonian $H$ through a block-encoding $U_H$, implement (an approximation of) $e^{i t H}$ for some time $t$.
    \item \textbf{Hamiltonian learning:}  given access to some Hamiltonian $H$ via its time evolution operator $U(t) = e^{i t H}$, output the matrix description of $H$ in the computational basis, minimizing the total evolution time.
    \item \textbf{Ground state preparation (GSP) of gapped Hamiltonians:}  given access to some Hamiltonian $H$ through a block-encoding $U_H$, implement (an approximation of) the ground state of $H$.
\end{itemize}

\begin{table}[!ht]
	\centering
	\begin{tabular}{c|c}
		Problem & Query lower bound  \\ \hline
		Quantum phase estimation (\cref{clm:QPE}) & $\Omega\left(1/\epsilon\right)$\\
  Entanglement entropy (\cref{clm:QEP}) & $\Omega\left(1/\sqrt{\Delta}\right)$\\
    Subset support verification (\cref{clm:SSV}) & $\Omega\left(\sqrt{|S|}\right)$\\
    Quantum amplitude estimation (\cref{cor:QAE}) & $\Omega\left(1/\epsilon\right)$\\
    Thermal state preparation (\cref{claim:QGS}) & $\Omega(\beta)$\\
   Hamiltonian simulation (\cref{claim:HS}) & $\Omega(t)$\\
      Hamiltonian learning (\cref{claim:learning}) & $\Omega(1/\epsilon)$\\
       GSP of gapped Hamiltonians (\cref{claim:GSP}) & $\Omega(1/\Delta)$
	\end{tabular}
	\caption{Obtained bounds for the query complexity of unitary property testing and related problems. All bounds hold for any $\mathsf{C}$-tester with $\mathsf{C} \subseteq \QMA(2)\slash\mathsf{qpoly}$ and are, with the exception of the entanglement entropy problem, shown to be tight (up to logarithmic factors) in~\cref{sec:app}.}
	\label{tbl:bounds}
\end{table}

\noindent \cref{tbl:bounds} summarises our obtained bounds for the above problems. All our lower bounds are obtained via~\cref{thm:lb_technique}, and upper bounds are given by known results or explicitly proven in the main text (see~\cref{sec:app}). With the exceptions of subset support verification and entanglement entropy, these bounds were already known to hold in literature~\cite{bessen2005lower,nayak1999quantum,berry2007efficient,kastoryano2023quantum,mande2023tight,huang2023learning,Lin2020nearoptimalground}.\footnote{The lower bound in Ref.~\cite{berry2007efficient} assumes sparse access to the entries of the local Hamiltonian, which is a stronger input model than the block-encoding setting we consider.} However, our work is (as far as the author is aware) the first to show that these bounds hold even in the presence of \emph{both} (quantum) proofs and advice.\footnote{Mande and de Wolf~\cite{mande2023tight} also give a lower bound on quantum phase estimation with advice, but the advice in their work refers to something different: the problem they consider generalises quantum phase estimation in the sense that you are no longer given the exact eigenstate but only a quantum state, called an advice state, which has some promised \textit{overlap} with it. In our case, advice refers to complexity classes that have access to some trusted string (or state) that may depend on the input length but not on the input itself. }

\subsection{Quantum oracle separations and implications to \texorpdfstring{$\QMA(2)$}{QMA(2)}}
One of the main takeaways of the results in~\cref{tbl:bounds} seems to be that whenever high precision is required in a black-box setting, neither quantum proofs nor advice seems to help in reducing the required query complexity. The complexity class $\QMA(k)$ is a generalisation of $\QMA$ where there are $k$ multiple \emph{uninteracting} provers, which are thus guaranteed to be unentangled with each other. It is known that $\QMA(k) = \QMA(2)$ for $2 \leq k \leq \poly(n)$~\cite{harrow2013testing}, but is generally believed that $\QMA \neq \QMA(2)$, since there are problems known to be in $\QMA(2)$ but not known to be in $\QMA$~\cite{liu2007quantum,aaronson2008power,beigi2008np,blier2009all}. However, it is still a major open question of how more powerful $\QMA(2)$ can be---the best upper bound we currently have is $\QMA(2) \subseteq \NEXP$, which follows by just guessing exponential-size classical descriptions of the two quantum proofs. So could $\QMA(2)$ even be equal to $\NEXP$, or even contain, say, $\PSPACE$? In order for this to be true, it needs to at least contain the class $\SBQP$ in the unrelativized world, where $\SBQP \subseteq \PP$ can be viewed as a variant of $\BQP$ with exponentially close completeness and soundness parameters.  However, our results easily imply a quantum oracle relative to which this is not true:
\begin{restatable}{theorem}{orsepMP}
There exists a quantum oracle $U $ such that 
\begin{align*}
    \QMA(2)^U \not\supset \SBQP^U.
\end{align*}
\label{thm:orsepMP}
\end{restatable}
The proof, given in~\cref{sec:oracle_sep}, considers any unary language $ L $ that is not contained in either $ \QMA(2) $ or $ \SBQP $, which must exist by a counting argument. We then encode $ L(x) $ as an eigenphase corresponding to the eigenstate $ \ket{0} $ in a diagonal unitary: the phase is $ 0 $ if $ x \notin L $ and exponentially close to $0$ if $ x \in L $. Running a one-bit quantum phase estimation circuit~\cite{kitaev1995quantum} with this unitary then yields an exponentially small probability of measuring $\ket{1}$ on the first qubit only when the eigenphase is non-zero. Therefore, we have that $ L $ is contained in $ \SBQP^U $, as it can distinguish exponentially small differences in the eigenphase making only a single controlled query to $U$.  For $ \QMA(2) $, the quantum oracle behaves almost like the identity operator for large input sizes. Hence, in these cases, $U$ can be replaced by the identity operator while preserving a sufficiently large completeness and soundness gap. This would contradict the assumption that we considered some $ L \notin \QMA(2) $.

Whilst this oracle separation is straightforward once one has~\cref{thm:lb_technique}, we believe it has theoretical merit: it implies that if $\QMA(2)$ would be able to solve very precise problems, for example the local Hamiltonian problem at exponentially precise precision (which is $\PSPACE$-complete~\cite{fefferman2016quantum}), it needs to be able to do so in a way that does not work in a quantum black-box setting.

Similarly, we can use the same idea to show a quantum oracle separation between $\QMA\slash\mathsf{qpoly}$ and $\SBQP$.\footnote{Contrary to what was claimed in an earlier version of this work, it is not clear if such an oracle separation holds if we consider both unentangled quantum proofs and advice, since we do not know whether $\QMA(2)\slash\mathsf{qpoly} = \ALL$~\cite{aaronson2005qma}.}
\begin{restatable}{theorem}{orsepadvice}
There exists a quantum oracle $U$ such that 
\begin{align*}
    \QMA\slash\mathsf{qpoly}^U \not\supset \SBQP^U.
\end{align*}
\label{thm:orsepadvice}
\end{restatable}

\subsection{Related work}
Classical property testing in a quantum setting was first studied by Buhrman, Fortnow, Newman, and Hein R{\"o}hrig~\cite{buhrman2008quantum}. There are two main techniques for proving lower bounds on the quantum query complexity of classical property testing: the polynomial method~\cite{beals2001quantum} and the adversary method~\cite{ambainis2000quantum}. The methods are generally incomparable: there exist problems in which the polynomial method provides a tight lower bound and the adversary method provably not~\cite{kutin2003quantum} and vice versa~\cite{ambainis2006polynomial}.\footnote{This does not hold if one considers the negative weights (or generalized) adversary method~\cite{hoyer2007negative}, which is tight for every decision problem~\cite{reichardt2009span}.}

The first study of unitary property testing by Wang~\cite{wang2011property} only includes upper bounds and raises the open question of finding a technique to show lower bounds for these kinds of problems.\footnote{The work also uses a more restricted query model: one starts with some arbitrary bipartite quantum state $\ket{\psi_{AB}}$, applies $U$ (or its inverse) to the $A$-subsystem of $\ket{\psi_{AB}}$ followed by a measurement.} In 2015 Belovs generalised the adversary method to include unitary property testing~\cite{belovs2015variations} and recently She and Yuen provided a lower bound method based on polynomials, which they also use to show lower bounds for $\QMA$-testers~\cite{she2022unitary}. Some other works that consider the query complexity of problems related to unitary property testing problems are Refs.~\cite{chen2023testing,mande2023tight,montanaro2013survey,aharonov2022quantum}.

\paragraph{Comparison to Ref.~\cite{wang2023quantum}.} We study many of the same problems as Wang and Zhang, but do so with a different technique: Wang and Zhang propose a sample-to-query lifting theorem, which translates a sample lower bound for any type of quantum state testing into a query lower bound for a corresponding type of unitary property testing. However, the obtained bounds using this method contain a sub-optimal inverse logarithmic factor, which seems to be inherent to their technique. The idea is to reduce unitary property testing problems to state discrimination, where the unitaries being tested block-encode the states from the state discrimination problem. The encoding consists of two steps: (1) copies of a state are used to approximately implement a quantum channel corresponding to time evolution under the density matrix as a Hamiltonian (known as density matrix exponentiation~\cite{lloyd2014quantum}), and (2) the density matrix exponentiation is leveraged within the quantum singular value transformation framework~\cite{gilyen2019quantum} to achieve a desired block-encoding (see~\cref{def:block_encoding}). Both steps are approximate, with the second incurring an additional logarithmic factor, leading to slightly suboptimal lower bounds. Moreover, their results are not shown to hold in the presence of quantum proofs and advice, which are two open problems raised in the discussion section of the paper. Using our technique we resolve both of these open problems for their obtained lower bounds. 

\subsection{Open questions}
Can one improve either the upper or lower bound for the entanglement entropy problem, such that both bounds match?\footnote{Following this work, Chen, Wang and Zhang improved the lower bound for the entanglement entropy to $\tilde{\Omega}(d/\Delta)$, simultaneously a dependence on the dimensionality $d$ as well as the precision $\Delta$~\cite{chen2024local}. However, as far as the author is aware, it is still open whether this lower bound can be matched by an upper bound.} Are there more unitary problems for which our method can prove (tight) lower bounds? Are there more examples of unitary property testing problems besides the quantum search problem from Ref.~\cite{aaronson2007quantum}, where one can show that quantum proofs provably reduce the quantum query complexity? What about quantum advice?

\paragraph{Acknowledgements.}
JW thanks Harry Buhrman, Alvaro Yanguez and Marco Túlio Quintino for useful discussions, and Amira Abbas, Jonas Helsen, Ronald de Wolf and Sebastian Zur for comments on an earlier draft.
Anonymous referees are also thanked for their feedback on earlier versions. JW was supported by the Dutch Ministry of Economic Affairs and Climate Policy (EZK), as part of the Quantum Delta NL programme. 

\section{Preliminaries}
\label{sec:prelim}

\subsection{Notation}
We write $[n]$ to denote the set of integers $\{1,2,\dots,n\}$. For a matrix $M$, we write $\eig(M)$ to denote the set of eigenvalues of $M$. For any $d \in \mathbb{N}$, $\mathbb{I}_d$ will be the $d$-dimensional identity matrix. If the dimension is clear from the context, we simply write $\mathbb{I}$. If $X$ is a finite set, $|X|$ is the size of $X$ and $\conv(X)$ the convex hull of the elements in $X$. We write $\poly(n)$ to denote an arbitrarily polynomially-bounded function. We denote $\mathbb{U}(d)$ for the unitary group of degree $d$ and $\mathcal{D}(d)$ for the set of all $d$-dimensional density matrices $\rho$ (i.e.~all complex $d$ by $d$ matrices that are PSD and have trace $1$). The operator, trace and diamond norm are denoted $\opnorm{\cdot}$, $\tnorm{\cdot}$ and $\dnorm{\cdot}$, respectively.

\subsection{Classes of unitary property testers}
We follow the conventions in the definitions by She and Yuen~\cite{she2022unitary}. For a fixed number of qubits $n$, let $\mathcal{P}_\textup{yes}$ and $ \mathcal{P}_\textup{no}$ (called the yes and no instances, respectively) be \emph{disjoint} subsets of all $n$-qubit unitary operators.\footnote{More generally, one would consider all $d$-dimensional unitaries for an arbitrary dimension $d$.} A tester for a unitary problem $\mathcal{P} = (\mathcal{P}_\textup{yes}, \mathcal{P}_\textup{no})$ is a quantum algorithm which given query access to $U \in \mathcal{P}_\textup{yes} \cap \mathcal{P}_\textup{no}) $ accepts with high probability when $U \in \mathcal{P}_\textup{yes}$ and accepts with low probability when $U \in \mathcal{P}_\textup{no}$. Hence, the tester can discriminate between the unitaries coming from the set of yes instances $\mathcal{P}_\textup{yes}$ and those from $\mathcal{P}_\textup{no}$. This is more general than the standard definition in the literature, where the no-instances are defined to be $\epsilon$-far from the set of yes instances in terms of some distance measure (see for example Ref.~\cite{montanaro2013survey}). However, for our purposes we only need the sets  $\mathcal{P}_\textup{yes}$ and $ \mathcal{P}_\textup{no}$ to be disjoint.

We can now build classes of testers similar to Ref.~\cite{she2022unitary}. First, Let us start with the simplest class of testers, which do not use proofs or advice.  W.l.o.g.,~one can always consider the input states to be pure since the acceptance probability of any mixed state is always upper bounded by a pure state via a convexity argument. The initial state of any quantum algorithm is always assumed to have some number of ancilla qubits all initialised in $\ket{0}$ and is sometimes supplied with an additional input state, depending on the class of tester considered.
\begin{definition}[$\BQP$-tester] Let $\mathcal{P} = (\mathcal{P}_\textup{yes},\mathcal{P}_\textup{no})$ be a unitary property on $n$ qubits. We say $\mathcal{P}$ has a $\BQP$-tester if there exists a quantum algorithm such that the following holds:
\begin{itemize}
    \item If $U \in \mathcal{P}_\textup{yes}$, the quantum algorithm makes queries to $U$ and accepts with probability $\geq 2/3$. 
    \item If $U \in \mathcal{P}_\textup{no}$, the quantum algorithm makes queries to $U$ and accepts with probability $\leq 1/3$.
\end{itemize}
\end{definition}
Note that there is no restriction on the number of queries the tester makes to $U$, as this is the quantity being characterized. Nor is there any restriction on the amount of space or time used by the $\BQP$-tester in this definition, which makes the ``$\mathsf{P}$'' in ``$\BQP$-tester'' somewhat awkward. Nonetheless, we adopt this notation to follow the convention in Ref.~\cite{she2022unitary} and to facilitate a direct connection to separations between actual complexity classes (\cref{sec:oracle_sep}).

Let us now add, possibly unentangled, quantum (resp.~classical) proofs to define  $\QMA$ (resp.~$\QCMA$) testers. 

\begin{definition}[$\QMA(k)$-tester] Let $\mathcal{P} = (\mathcal{P}_\textup{yes},\mathcal{P}_\textup{no})$ be a unitary property on $n$ qubits. We say $\mathcal{P}$ has a $\QMA(k)$-tester if the following holds:
\begin{itemize}
    \item If $U \in \mathcal{P}_\textup{yes}$, then there exists $k$ quantum states $\{ \ket{\xi_i} \}$ such that on input $\ket{\xi_1} \dots \ket{\xi_k}$ the algorithm makes queries to $U$ and accepts with probability $\geq 2/3$. 
    \item If $U \in \mathcal{P}_\textup{no}$, then for all $k$ quantum states $\{ \ket{\xi_i} \}$, the algorithm acting on $\ket{\xi_1} \dots \ket{\xi_k}$  makes queries to $U$ and accepts with probability $\leq 1/3$.
\end{itemize}
If $k = 1$, we abbreviate to a $\QMA$-tester. 
\end{definition}
\begin{subdefinition}[$\QCMA$-tester] This has the same definition as the $\QMA$-tester, but where the promises hold with respect to computational basis states $\ket{y}$ as input states.
\end{subdefinition}
By the result of Harrow and Montanaro, it would suffice to consider $\QMA(2)$-testers to cover all of $\QMA(\poly(n))$~\cite{harrow2013testing}.

The definitions of testers of unitary properties in Ref.~\cite{she2022unitary} do not include classes which allow for advice. A technical difficulty that arises when one wants to include advice states is that we need to specify a notion of input length on which the advice may depend. in Ref.~\cite{she2022unitary}, no such restrictions were necessary, and all classes of testers were defined in a way that did not in any way depend on any notion of input size $n$. We will make the choice that $n$ is given by the number of qubits the unitary acts on. Note that in the case where the property is also parametrised by some parameter, this is not restrictive in the advice setting since every parameter setting is a different unitary property and hence allows for a different choice of advice strings. To make this point clear, take the example of a unitary property testing formulation of quantum phase estimation with a variable number of qubits $n$ but a fixed sequence of known eigenstates $\ket{\psi_n}$, for example $\ket{0^n}$: here one has to determine if an unknown $n$-qubit unitary $U$ comes from $\mathcal{P}_\textup{yes} = \{ U : U \ket{0^n} = \ket{0^n} \}$ or $\mathcal{P}_\textup{no} = \{ U : U \ket{0^n} =e^{2 \pi i \theta} \ket{0^n}, \epsilon \leq  \theta \leq 1/2 \}$ for some fixed $\epsilon >0$. In other words, promised that $U$ has an eigenstate $\ket{\psi}$, determine if its eigenphase is $0$ or $\geq \epsilon$. In this case, one could also parametrise $\epsilon$ as a function of some $m$, i.e.,~$\epsilon = \epsilon(m)$. In this case, for a fixed choice of $\epsilon$, the advice should be identical for each fixed $n$. However, even for fixed values of $n$, the advice can be different for different values of $\epsilon$ as each is a new property testing problem. Having set our choice of what the input size $n$ means, we will now state our definitions of the testers.

\begin{definition}[$\BQP\slash\mathsf{qpoly}$-tester] Let $\mathcal{P} = (\mathcal{P}_\textup{yes},\mathcal{P}_\textup{no})$ be a unitary property on $n$ qubits. We say $\mathcal{P}$ has a $\BQP\slash\mathsf{qpoly}$-tester if there exists a collection of quantum advice states $\{\ket{\psi_n}\}$, and a quantum algorithm such that the following holds:
\begin{itemize}
     \item If $U \in \mathcal{P}_\textup{yes}$, on input $\ket{\psi_n}$ the quantum algorithm makes queries to $U$ and accepts with probability $\geq 2/3$. 
    \item If $U \in \mathcal{P}_\textup{no}$, on input $\ket{\psi_n}$ the quantum algorithm makes queries to $U$ and accepts with probability $\leq 1/3$.
\end{itemize}
\label{def:BQPqpoly}
\end{definition}
\begin{subdefinition}[$\BQP\slash\mathsf{poly}$-tester] This has the same definition as the $\BQP\slash\mathsf{qpoly}$-tester but where the advice states are computational basis states $\ket{y_n}$.
\end{subdefinition}

We can also combine proofs and advice to arrive at even more powerful classes of testers.

\begin{definition}[$\QMA(k)\slash\mathsf{qpoly}$-tester]
Let $\mathcal{P} = (\mathcal{P}_\textup{yes},\mathcal{P}_\textup{no})$ be a unitary property on $n$ qubits. We say $\mathcal{P}$ has a $\QMA(k)\slash\mathsf{qpoly}$-tester if there exists a collection of quantum advice states $\{\ket{\psi_n}\}$, where $\ket{\psi_n} \in (\mathbb{C}^{2})^{\otimes \poly(n)}$, and a quantum algorithm such that the following holds:
\begin{itemize}
    \item If $U \in \mathcal{P}_\textup{yes}$, then there exists $k$ quantum states $\{ \ket{\xi_i} \}$, such that on input $\ket{\psi_n} \ket{\xi_1} \dots \ket{\xi_k}$ the algorithm makes queries to $U$ and accepts with probability $\geq 2/3$. 
    \item If $U \in \mathcal{P}_\textup{no}$, then for all $k$ quantum states $\{ \ket{\xi_i} \}$,  the algorithm acting on $\ket{\psi_n} \ket{\xi_1} \dots \ket{\xi_k}$  makes queries to $U$ and accepts with probability $\leq 1/3$.
\end{itemize}
\label{def:QMAkqpoly}
\end{definition}
\noindent In this case it still holds that $\QMA(\poly(n))\slash\mathsf{qpoly} = \QMA(2)\slash\mathsf{qpoly}$, since the $\QMA(2)=\QMA(\poly(n))$ result of Ref.~\cite{harrow2013testing} relies on being able to boost the completeness parameter $c$ to be close to $1$, which can be done by error reduction (if there is an advice state which works, then one is allowed to use as many copies as one wants, since it still satisfies the criterion that it only depends on the size of the input). One can also define $\QCMA\slash\mathsf{qpoly}$-, $\QCMA\slash\mathsf{poly}$- and $\QCMA\slash\mathsf{poly}$- and $\QMA\slash\mathsf{poly}$-testers by modifying the promises such that they hold with respect to basis states on either the proof, the advice or both. All classes of testers discussed so far are contained in $\QMA(2)\slash\mathsf{qpoly}$.

There is one more class we would like to introduce, which is fundamentally different from all classes of testers discussed so far in the sense that it allows for an exponentially small gap between the completeness and soundness parameters.

\begin{definition}[$\SBQP$-tester]
Let $\mathcal{P} = (\mathcal{P}_\textup{yes},\mathcal{P}_\textup{no})$ be a unitary property on $n$ qubits. We say $\mathcal{P}$ has a $\SBQP$-tester if there exists a quantum algorithm and a polynomial $p(n)$ such that the following holds:
\begin{itemize}
   \item If $U \in \mathcal{P}_\textup{yes}$, the quantum algorithm makes queries to $U$ and accepts with probability $\geq 2^{-p(n)}$. 
    \item If $U \in \mathcal{P}_\textup{no}$, the quantum algorithm makes queries to $U$ and accepts with probability $\leq 2^{-p(n)-1}$.
\end{itemize}
\label{def:SBQP}
\end{definition}

\subsection{Quantum information}
In this work, all quantum systems will be finite-dimensional. For two $d$-dimensional quantum states described by density operators $\rho_1,\rho_2 \in \mathcal{D}(d)$, the trace distance is given by
\begin{align}
    \frac{1}{2}\tnorm{\rho_1-\rho_2} = \frac{1}{2} \tr\left[\sqrt{(\rho_1-\rho_2)^\dagger(\rho_1-\rho_2)}\right].
\end{align}
Equivalently, the trace distance can be defined in a variational form as
\begin{align}
    \frac{1}{2}\tnorm{\rho_1-\rho_2} = \sup_{0 \leq P \leq \mathbb{I}} \tr\left[P(\rho_1-\rho_2)\right],
\end{align}
where $P$ is a positive operator. Hence, the trace distance characterizes the largest possible bias with which one can distinguish between states $\rho_1$ and $\rho_2$. A quantum channel $\mathcal{E}$ is a completely positive trace-preserving map between two density operators $\rho_1 \in \mathcal{D}(d_1)$ and  $\rho_1 \in \mathcal{D}(d_2)$, where the dimensions $d_1$ and $d_2$ do not have to be the same. A special type of quantum channel is a unitary channel, which is a channel which applies to a $d$-dimensional density matrix $\rho$ a mapping 
\begin{align}
   \mathcal{U}(U) : \rho \rightarrow U \rho U^\dagger,
   \label{eq:def_unitary_channel}
\end{align}
for some unitary $U \in \mathbb{U}(d)$. Given a unitary $U \in \mathbb{U}(d)$, we will also write $\mathcal{U}(U)$ for the unitary channel implementing a specified $U$ as per~\cref{eq:def_unitary_channel}.  The distance between two quantum channels can be characterized using the so-called diamond distance, which is induced by the diamond norm.

\begin{definition}[Diamond distance for quantum channels] Let $d_a,d_b$ be the dimensions of two finite-dimensional Hilbert spaces. The diamond-norm distance, denoted as $\frac{1}{2}\dnorm{\mathcal{E}_1 - \mathcal{E}_2}$, between two quantum channels $\mathcal{E}_1,\mathcal{E}_2 : \mathcal{D}(d_a) \rightarrow \mathcal{D}(d_b)$ is given by 
\begin{align*}
     \frac{1}{2 }\dnorm{\mathcal{E}_1 - \mathcal{E}_2} = \frac{1}{2} \max_{\tr(\rho) = 1, \rho \succeq 0} \tnorm{(\mathcal{E} \otimes \mathbb{I}_{d_a}) (\rho) - (\mathcal{E}_2 \otimes  \mathbb{I}_{d_a}) (\rho)}.
\end{align*}
If $\mathcal{E}_1$ and $\mathcal{E}_2$ are the corresponding channels of unitaries $U_1$ and $U_2$, respectively, then
\begin{align*}
    \frac{1}{2} \dnorm{\mathcal{U}(U_1)-\mathcal{U}(U_2)} = \frac{1}{2} \max_{\tr(\rho) = 1, \rho \succeq 0} \tnorm{ U_1 \rho U_1^\dagger -  U_2\rho U_2^\dagger }.
\end{align*}
\label{def:diamond_distance_channels}
\end{definition}
The diamond distance between two channels $\mathcal{E}_1$ and $\mathcal{E}_2$ precisely characterises the one-shot distinguishability between the two channels, i.e.~the success probability of correctly identifying a channel (when they are given with equal probability) is given by
\begin{align}
    p_\text{succ} = \frac{1}{2} + \frac{1}{2} \dnorm{\mathcal{E}_1 - \mathcal{E}_2}.
\end{align}
We will need the following two useful properties of the diamond distance, which can be found in (or easily derived from) Propositions 3.44 and 3.48 in Ref.~\cite{watrous_2018}.

\begin{lemma}[Two properties of the diamond distance]
Let $d_a,d_b,d_c \geq 1$ be arbitrary dimensions of the complex Hilbert spaces. The diamond distance satisfies the following two properties:
\begin{enumerate}
    \item Unitary invariance. For all quantum channels $\mathcal{E}_1$,$\mathcal{E}_2 : \mathcal{D}(d_a) \rightarrow \mathcal{D}(d_b)$, for all $U \in \mathbb{U}(d_a)$ and $V \in \mathbb{U}(d_B)$ we have 
    \begin{align*}
        \dnorm{\mathcal{E}_1-\mathcal{E}_2  } = \dnorm{\mathcal{E}^{'}_1-\mathcal{E}^{'}_2  },
    \end{align*}
    where $\mathcal{E}^{'}_i = V \mathcal{E}_i (U \rho U^\dagger) V^\dagger  $ for a $\rho \in \mathcal{D}(d_a)$ , for $i \in \{1,2\}$.
    \item Hybrid argument. For all quantum channels $\mathcal{E}_1,\mathcal{E}^{'}_1 : \mathcal{D}(d_a) \rightarrow \mathcal{D}(d_b)$, $\mathcal{E}_2, \mathcal{E}^{'}_2: \mathcal{D}(d_b) \rightarrow \mathcal{D}(d_c)$, we have that
    \begin{align*}
        \dnorm{\mathcal{E}_1\mathcal{E}_2-\mathcal{E}^{'}_1 \mathcal{E}^{'}_2  } \leq + \dnorm{\mathcal{E}_1-\mathcal{E}^{'}_1 } + \dnorm{\mathcal{E}_2 -\mathcal{E}^{'}_2  }.
    \end{align*}
\end{enumerate}
\label{lem:prop_diamondnorm}
\end{lemma}

\section{Lower bounds by unitary channel discrimination}
\label{sec:technique}
In this section we will prove~\cref{thm:lb_technique}. First, we will set up some more preliminaries regarding unitary channel discrimination. For unitary channels, the definition of diamond norm allows for a more easily computable expression, as shown in the following lemma.
\begin{lemma}[Adapted from Ref.~\cite{ziman2010single}] Let $U_1,U_2 \in \mathbb{U}(d)$ such that $0 \notin \conv(\eig(U_1^\dagger U_2))$. Then
\begin{align}
    \frac{1}{2}\dnorm{\mathcal{U}(U_1) - \mathcal{U}(U_2)} = \sqrt{1-D^2},
\end{align}
with 
\begin{align}
    D = \frac{1}{2} \min_{k,l} \abs{e^{i \theta_k} + e^{i \theta_l}} 
\end{align}
where $ e^{i \theta_k},e^{i \theta_l} \in \eig(U_1^\dagger U_2)$.
\label{lem:norm_conv1}
\end{lemma} 
The quantity $D$ has a nice geometrical interpretation: if $0 \notin \conv(\eig(U_1^\dagger U_2))$, then it is precisely the distance between the convex hull of the eigenvalues of $U_1^\dagger U_2$ and the origin. To determine whether $0 \notin \conv(\eig(U_1^\dagger U_2))$, it is useful to consider the \textit{spectral arc-length} of $U$, given by the following definition.
\begin{definition}[Spectral arc-length] The spectral arc-length $\Theta(U)$ of a unitary $U \in \mathbb{U}(d)$ is given by the length of the shortest compact interval $I \subset \mathbb{R}$ such that the corresponding segment $e^{i I}$ contains the spectrum of $U$. 
\label{def:sal}
\end{definition}
\noindent It is easy to show that $0 \notin \conv(\eig(U)$ if and only if $\Theta(U) < \pi$~\cite{wolf2022mathematical} (see also~\cref{fig:complex_circle}). 

There are more ways to characterise the diamond distance for unitaries, which might be more convenient for some choices of unitaries as we try to prove lower bounds later down the road. Since we are interested in the asymptotic scaling of our lower bounds, it will sometimes be convenient to use the (global-phase shifted) operator norm difference between the two unitaries, as it is equivalent to the diamond norm up to a (inverse) factor of (at most) two. 

\begin{lemma}[[Adapted from~\cite{haah2023query}] Let $U_1,U_2 \in \mathbb{U}(d)$  be unitaries. Then we have that
\begin{align*}
     \frac{1}{2} \dnorm{\mathcal{U}(U_1)-\mathcal{U}(U_2)} \leq  \min_{\phi} \opnorm{e^{i \phi} U_1 - U_2} \leq \dnorm{\mathcal{U}(U_1)-\mathcal{U}(U_2)}.
\end{align*}
\label{lem:norm_conv2}
\end{lemma}
As~\cref{lem:norm_conv2} involves a minimisation over a global phase factor $e^{i \phi}$, using the vanilla operator norm distance between two unitaries $U_1$ and $U_2$ would also yield an upper bound. Moreover, one can also show that a sufficient upper bound in operator norm implies the condition that $0 \notin \conv(\eig(U)$ is immediately satisfied, as illustrated by the following lemma.

\begin{lemma} Let $U_1,U_2 \in \mathbb{U}(d)$  be unitaries for which $\opnorm{ U_1 - U_2} \leq 2\sin(\frac{\pi-\delta}{4})$ for some small $0 < \delta <\pi$. Then it holds  that $0 \notin \conv(\eig(U_1^\dagger U_2))$. For $\delta = 0.04$, we have $2\sin(\frac{\pi-\delta}{4}) = 1.4$.
\label{lem:op_distance_convhull}
\end{lemma}
\begin{proof}
By unitary invariance of the operator norm it must hold that
\begin{align*}
    \opnorm{ U_1 - U_2} = \opnorm{\mathbb{I} - U_1^\dagger U_2} = \max_{\theta_j : e^{i \theta_j} \in \eig(U_1^\dagger U_2)} \abs{1-e^{i \theta_j}}.
\end{align*}
Since all $e^{i\theta_j}$ are points on the unit circle in the complex plane, we can assume w.l.o.g.~that $\theta_j \in [0,2\pi]$ for all $j \in [d]$. Thus, we get the condition
\begin{align*}
    \max_{\theta_j : e^{i \theta_j} \in \eig(U_1^\dagger U_2)} 2 \abs{\sin \left(\frac{\theta_j}{2}\right)} \leq 2\sin(\frac{\pi-\delta}{4}),
\end{align*}
which means that $0 \leq \theta_j \leq (\pi-\delta)/2$ or $2 \pi - (\pi-\delta)/2 \leq \theta_j \leq 2\pi$ for all $j \in [d]$. Hence, all eigenvalues of $U^\dagger_1 U_2$ must lie within the segment $e^{i I}$ on the unit circle in the complex plane given by the bounded and closed interval $I = [-(\pi-\delta)/2,(\pi-\delta)/2]$, which has length $\pi-\delta$ and is compact by the Heine-Borel theorem. Therefore,~\cref{def:sal} combined with the assumption $0<\delta < \pi $ implies that $\Theta(U_1^\dagger U_2)  < \pi$, which means that $0 \notin \conv(\eig(U_1^\dagger U_2))$.
\end{proof}
\begin{figure}
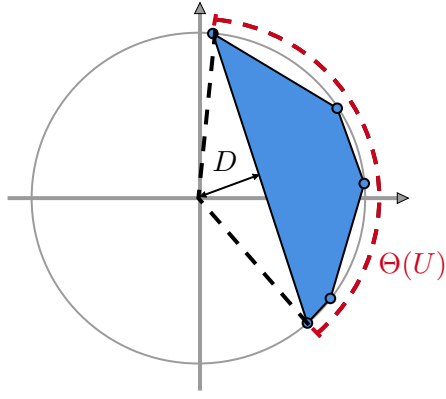

    \centering
    \include{tikz_figure}
    \caption{Geometrical interpretation in the complex plane of the spectrum of some $U \in \mathbb{U}(5)$. Since $U$ is unitary, all its eigenvalues (blue dots) lie on the unit circle in the complex plane. Since the spectral arc-length (indicated by the red dashed line) satisfies $\Theta(U) < \pi$, we have that $0$ is not in the convex hull of the eigenvalues of $U$ (blue shaded area). The shortest distance from the origin to the convex hull of the spectrum of $U$ is $D$.}
    \label{fig:complex_circle}
\end{figure}

Now that we have established ways to evaluate the diamond distance between two unitaries, we will proceed by showing that adding a well-chosen global phase to one of the unitaries ensures that access to the inverse, controlled access, or a combination of both, does not increase the ability to discriminate between the two. To show this, first observe that the condition of $0$ not being in the convex hull of the eigenvalues of $U_1^\dagger U_2$ is invariant to adding a global phase to one of the unitaries.

\begin{lemma}Let $U_1,U_2 \in \mathbb{U}(d)$ such that $0 \notin \conv(\eig(U_1^\dagger U_2))$. Let $U_2^{\theta} = e^{i \theta} U_2 $. Then for all $\theta \in [0,2\pi)$ we have $0 \notin \conv(\eig(U_1^\dagger U_2^{\theta}))$.
\label{lem:rot_0convhull}
\end{lemma}
\begin{proof}
    This follows directly from the invariance of $D$ under rotations, since  $\eig(U_1^\dagger U_2^{\theta}) = \{e^{i\theta} e^{i \theta_l} | e^{i \theta_l} \in \eig(U_1^\dagger U_2)\}$ and 
    \begin{align*}
        \frac{1}{2} \min_{k,l} \abs{e^{i\theta}e^{i \theta_k} + e^{i\theta}e^{i \theta_l}} =  \frac{1}{2} \min_{k,l} \abs{e^{i\theta}}\abs{e^{i \theta_k} + e^{i \theta_l}} =  \frac{1}{2} \min_{k,l}\abs{e^{i \theta_k} + e^{i \theta_l}}
    \end{align*}
    for any $\theta \in [0,2\pi)$.
\end{proof}
\noindent The next lemma proves the aforementioned claim that given some $U_1,U_2$ we can always pick a global phase such that the diamond distance between $U_1$ and $U_2$ precisely characterises the distance in the controlled, inverse, or controlled inverse setting.

\begin{lemma}[Diamond distance for different query types] Let $U_1,U_2 \in \mathbb{U}(d)$ such that $0 \notin \conv(\eig(U_1^\dagger U_2))$, and let $U_2^{\theta} = e^{i \theta} U_2$ for some $\theta \in [0,2\pi)$. Then there exists a choice of $\theta$ such that for all choices $(\tilde{U}_1,\tilde{U}_2^{\theta}) \in \{({U}_1,{U}_2^{\theta}), (U_1^\dagger,U_2^{\theta,\dagger}), (cU_1,cU_2^{\theta}), (cU_1^\dagger,cU_2^{\theta,\dagger}) \}$  we have
\begin{align*}
    \frac{1}{2}\dnorm{\mathcal{U}(\tilde{U}_1) -\mathcal{U}(\tilde{U}_2^{\theta})} = \frac{1}{2} \dnorm{\mathcal{U}(U_1) -\mathcal{U}(U_2^{\theta})}.
\end{align*}
\label{lem:dddqt}
\end{lemma}
\begin{proof}
By~\cref{lem:rot_0convhull}, we must have that  $0 \notin \conv(\eig(U_1^\dagger U_2^{\theta}))$ for any $\theta \in [0,2\pi)$. We consider the separate cases using~\cref{lem:norm_conv1}. The $({U}_1,{U}_2^{\theta})$-case is trivially true.
\paragraph{Inverse queries.} By Theorem 1.3.22 in Ref.~\cite{horn2012matrix} we have $\eig((U_1^\dagger)^\dagger U_2^{\theta,\dagger}  ) = \eig( U_2^{\theta,\dagger} (U_1^\dagger)^\dagger ) = \eig( (U_1^\dagger U_2^{\theta})^\dagger) $. Therefore, for all $\theta \in [0, 2\pi)$ it holds that
\begin{align*}
        \frac{1}{2} \dnorm{\mathcal{U}(U_1^\dagger) -\mathcal{U}(U_2^{\theta,\dagger})} 
        &=  \sqrt{1-\left(\frac{1}{2} \min_{k,l} \abs{ e^{-i(\theta+\theta_k)} + e^{-i(\theta+\theta_l)} }\right)^2} \\
        &= \sqrt{1-\left(\frac{1}{2} \min_{k,l} \abs{ e^{i(\theta+\theta_k)} + e^{i(\theta+\theta_l)} } \right)^2}\\
        &=\frac{1}{2}\dnorm{\mathcal{U}(U_1) -\mathcal{U}(U_2^{\theta})}.
\end{align*}

\paragraph{Controlled queries.} Note that we have that
\begin{align*}
    \eig(c U_1^\dagger cU_2^{\theta}  ) = \eig(\ket{0}\bra{0} \otimes \mathbb{I} + \ket{1}\bra{1} \otimes U_1^\dagger U_2^\theta) =  \eig(U_1^\dagger U_2^\theta) \cup \{1\}.
\end{align*}
For all $U_1,U_2$ there exists a $\theta$ such that $1 \in \eig(U_1^\dagger U_2^\theta)$, simply by taking any eigenvalue $e^{i \theta_l} \in \eig(U_1^\dagger U_2)$ and letting $\theta = 2\pi - \theta_l$. Hence, for this choice of $\theta$ we have
\begin{align*}
        \frac{1}{2}\dnorm{\mathcal{U}(cU_1)-\mathcal{U}(cU_2^{\theta})} = \frac{1}{2} \dnorm{\mathcal{U}(U_1)-\mathcal{U}(U_2^{\theta})}.
    \end{align*}

\paragraph{Controlled queries to the inverse.} This follows directly by combining the two arguments above, and holds for the same values of $\theta$.
\end{proof}
We restate~\cref{thm:lb_technique} from~\cref{sec:int} and give the proof. 
\lowb*
\begin{proof}
We omit the dependence of $\theta$ in subsequent notation and assume that it is chosen according to~\cref{lem:dddqt}. W.l.o.g. we can write the circuit $V^{U_i,(T)}$ which makes $T$ queries to $\tilde{U_i}^t \in \{U_i,U_i^\dagger, cU_i, cU_i^\dagger , \}$, $i \in \{1,2\}$, as
\begin{align*}
    V^{U_i,(T)}  =  V_T \tilde{U}_i^T V_{T-1} \tilde{U}^{T-1}_i \dots V_1 \tilde{U}^1_i V_0 .
\end{align*}
 We want to measure a designated output qubit of $V^{U_i,(T)}$ in the computational basis, which should output $1$ if we have $U_1$ and output $0$ if we have $U_2$. 
 
 By~\cref{def:diamond_distance_channels} and the operational interpretation of trace distance, we then have that for any input state $\ket{\psi_\textup{init}}$
\begin{align*}
    \abs{\Pr[V^{U_1,(T)} (\ket{\psi_\textup{init}}) \text{ outputs } 1] - \Pr[V^{U_2,(T)}(\ket{\psi_\textup{init}}) \text{ outputs } 1]} \leq \frac{1}{2}\dnorm{\mathcal{U}(V^{U_1,(T)})- \mathcal{U}(V^{U_2,(T)})}.
\end{align*}
Let us show by induction on the number of queries that
\begin{align*}
    \dnorm{\mathcal{U}(V^{U_1,(T)})- \mathcal{U}(V^{U_2,(T)})} \leq 2 T \epsilon.
\end{align*}
For a single query, we have
\begin{align}
    \dnorm{\mathcal{U}(V^{U_1,(1)})- \mathcal{U}(V^{U_2,(1)})} & = \dnorm{\mathcal{U}(V_1 \tilde{U}^1_1 V_0) -  \mathcal{U}(V_1 \tilde{U}^1_2 V_0)}
    \label{eq:sq_1}\\
    &= \dnorm{\mathcal{U}(\tilde{U}^1_1)-\mathcal{U}(\tilde{U}^1_2)}\label{eq:sq_2}\\
    &=  \dnorm{\mathcal{U}(U_1)-\mathcal{U}(U_2)}\label{eq:sq_3}\\
    &\leq 2 \epsilon,\label{eq:sq_4}
\end{align}
using the unitary invariance of the diamond norm (\cref{lem:prop_diamondnorm}, property 1) in going from line~\cref{eq:sq_1} to line~\cref{eq:sq_2} and~\cref{lem:dddqt} from going to line~\cref{eq:sq_2} to~\cref{eq:sq_3}. For $T$ queries, assuming the induction hypothesis on $T-1$ queries, we find

\begin{align}
\dnorm{\mathcal{U}(V^{U_1,(T)})- \mathcal{U}(V^{U_2,(T)})}
&=  \dnorm{ \mathcal{U}(V_T \tilde{U}_1^{T} V^{U_1,(T-1)}) - \mathcal{U}(V_T \tilde{U}_2^{T} V^{U_2,(T-1)})} \label{eq:Tq_1} \\
&\leq 
\begin{aligned}[t]
&\dnorm{ \mathcal{U}(V_T \tilde{U}_1^{T}) - \mathcal{U}(V_T \tilde{U}_2^{T} )} \\
&\quad + \dnorm{\mathcal{U}(V^{U_1,(T-1)}) - \mathcal{U}(V^{U_2,(T-1)} )} 
\end{aligned} \label{eq:Tq_2} \\
&\leq  \dnorm{\mathcal{U}( V_T \tilde{U}_1^{T}) - \mathcal{U}(V_T \tilde{U}_2^{T} )} + 2(T-1)\epsilon  \label{eq:Tq_3} \\
&= \dnorm{\mathcal{U}(\tilde{U}_1^{T}) - \mathcal{U}(\tilde{U}_2^{T} )} + 2(T-1)\epsilon  \label{eq:Tq_4} \\
&= \dnorm{\mathcal{U}(U_1) - \mathcal{U}(U_2)} + 2(T-1)\epsilon  \label{eq:Tq_5} \\
&\leq 2T\epsilon,  \label{eq:Tq_6}
\end{align}
where the bound holds for any $\ket{\psi_\textup{init}}$. Here we used~\cref{lem:prop_diamondnorm}, property 2, in going~\cref{eq:Tq_1} to~\cref{eq:Tq_2}, the induction hypothesis in going from~\cref{eq:Tq_2} to~\cref{eq:Tq_3}, the unitary invariance of the diamond norm from~\cref{eq:Tq_3} to~\cref{eq:Tq_4} and~\cref{lem:dddqt} from~\cref{eq:Tq_4} to~\cref{eq:Tq_5}. Therefore, in order to have
\begin{align*}
    \abs{\Pr[V^{U_1,(T)} (\ket{\psi_\textup{init}}) \text{ outputs } 1] - \Pr[V^{U_2,(T)}(\ket{\psi_\textup{init}}) \text{ outputs } 1]} \geq \frac{2}{3}-\frac{1}{3} = \frac{1}{3},
\end{align*}
we require
\begin{align*}
    \frac{1}{3} \leq \frac{1}{3}\dnorm{\mathcal{U}(V^{U_1,(T)})- \mathcal{U}(V^{U_2,(T)})} \leq T \epsilon,
\end{align*}
which implies 
\begin{align}
    T \geq \frac{1}{3 \epsilon}.
    \label{eq:thm_bound}
\end{align}
We now have to show that~\cref{eq:thm_bound} holds for $\QMA(2)\slash\mathsf{qpoly}$-testers. Let $\mathcal{P} = (\mathcal{P}_\text{yes},\mathcal{P}_\text{yes})$, where $\mathcal{P}_\text{yes} = \{U_1^1,U_1^2,\dots\}$ and $\mathcal{P}_\text{no} = \{U_2^1,U_2^2,\dots\}$, such that for each $n \in \mathbb{N}$ we have $\frac{1}{2}\dnorm{\mathcal{U}(U_1^n)-\mathcal{U}(U_2^n)} \leq \frac{1}{3T}$. Now suppose the above corollary is false, then there must exist an input quantum state $\ket{\psi_\textup{input}} = \ket{\psi_n} \ket{\xi_1} \dots \ket{\xi_k}$ such that a $T'$-query quantum algorithm, $T' < T$, starting in the initial state $\ket{\psi_{\textup{init}}} = \ket{0}^{\otimes \poly(n)} \ket{\psi_\textup{input}}$, which makes queries to $U$ can decide with probability $> 2/3$ whether $U \in \mathcal{P}_\text{yes}$ or $U \in \mathcal{P}_\text{no}$ (in the no-case the soundness property holds for all states, so also the one that is used in the yes-case). This contradicts the bound of~\cref{eq:thm_bound}, which holds for any $\ket{\psi_{\textup{init}}}$.
\end{proof}

\begin{remark} \cref{thm:lb_technique} also holds when we replace the diamond norm condition on $U_1$ and $U_2$ by
\begin{align}
   \min_{\phi}\opnorm{e^{i \phi} U_1 - U_2} \leq \opnorm{ U_1 - U_2} \leq \epsilon,
\end{align}
since we show that
\begin{align*}
    \dnorm{\mathcal{U}(V^{U_1})- \mathcal{U}(V^{U_2})} \leq T \dnorm{\mathcal{U}(U_1)-\mathcal{U}(U_2)},
\end{align*}
where 
\begin{align*}
    \frac{1}{2}\dnorm{\mathcal{U}(U_1)-\mathcal{U}(U_2)} \leq \min_{\phi}\opnorm{e^{i \phi} U_1 - U_2},
\end{align*}
by~\cref{lem:norm_conv2}.
\end{remark}
With~\cref{thm:lb_technique} in hand, our proposed lower bound technique is now straightforward: given a unitary property $\mathcal{P} = (\mathcal{P}_\textup{yes},\mathcal{P}_\textup{no})$, one picks two unitaries $U_1$, $U_2$ such that $U_1 \in \mathcal{P}_\textup{yes}$ and $U_2 \in \mathcal{P}_\textup{no}$; this implies that any unitary property tester for $\mathcal{P}$ implies a distinguisher for $U_1$ and $U_2$. One evaluates the diamond distance (or an upper bound thereof) between $U_1$ and $U_2$ using~\cref{lem:norm_conv1,lem:norm_conv2} or any other suitable method, and verifies the property of $0 \notin \conv(\eig(U_1^\dagger U_2))$, possibly with the help of~\cref{lem:rot_0convhull}. The lower bound on the channel discrimination query complexity of $U_1,U_2$ follows then from~\cref{thm:lb_technique}, which directly implies a lower bound on the query complexity of the unitary property tester.

\section{Applications}
\label{sec:app}
In this section we will apply the lower bound method of~\cref{sec:technique} to some problems in unitary property testing as well as some other unitary problems, showcasing the fact that the technique is very simple to use and (for all but one problem) leads to optimal lower bounds. 
\subsection{Unitary property testing}
\subsubsection{Quantum phase estimation}
In quantum phase estimation, one is given a unitary $U$ and an eigenstate $\ket{\psi}$, and the task is to determine the eigenphase of $U$ corresponding to $\ket{\psi}$ up to some precision $\epsilon$ with probability $\geq 2/3$. A lower bound can be obtained by reducing the quantum counting problem to the amplitude estimation problem, which is then reduced to the phase estimation problem. From this, the lower bound of $\Omega(1/\epsilon)$ follows from the lower bound for quantum counting given in Ref.~\cite{nayak1999quantum}. However, our method allows us to prove the same lower bound in only a couple of lines, and shows that it holds even in the presence of proofs and advice.

\begin{claim}[Lower bound for quantum phase estimation] Let $U \in \mathbb{U}(d)$ and $\ket{\psi}$ a quantum state such that $U \ket{\psi} = e^{2\pi i \theta} \ket{\psi}$, for some $\theta \in [0,1)$. Suppose that either (i) $\theta \geq b$ or (ii) $\theta \leq a$, with $b-a = \epsilon$. Then any $\mathsf{C}$-tester, where $\mathsf{C} \subseteq \QMA(2)\slash\mathsf{qpoly}$, that decides whether (i) or (ii) holds with probability $\geq 2/3$ has to make at least
\begin{align*}
    \Omega(1/\epsilon)
\end{align*}
(controlled) queries to $U$ (or its inverse).
\label{clm:QPE}
\end{claim}
\begin{proof}
    Let $U_1 = \mathbb{I}$ and $U_2 = e^{ 2i \pi \epsilon} \ketbra{0} + \ketbra{1} $, with $\epsilon >0$. Note $\ket{\psi} = \ket{0}$ is an eigenstate of both $U_1$ and $U_2$ with eigenphases $\theta_1 = 0$ and $\theta_2 = \epsilon$, respectively. Hence, any algorithm that decides whether some $U$ with eigenstate $\ket{\psi}$ has an eigenphase $\theta \leq a:=0$ or $\geq b:=\epsilon$ can discriminate $U_1$ from $U_2$. We have $
        \eig(U_1^\dagger U_2) = \{e^{2i \pi \epsilon},1\}$,
    which means $0 \notin \conv(\eig(U_1^\dagger U_2)$. Using
$D = \frac{1}{2}\abs{1 + e^{ 2 i \pi \epsilon}} $ we find
    \begin{align*}
    \frac{1}{2}\dnorm{\mathcal{U}(U_1)-\mathcal{U}(U_2)} = \sqrt{1-\frac{1}{4}\abs{1 + e^{ 2 i \pi \epsilon}}^2} = \abs{\sin(\pi \epsilon)} \leq  \pi \epsilon
\end{align*}
for $\epsilon >0$. Hence, by~\cref{thm:lb_technique} we find 
        $T \geq \Omega\left(1/\epsilon\right) $.
\end{proof}
This matches the well-known upper bound by Kitaev~\cite{kitaev1995quantum}.

\subsubsection{Entanglement entropy}
In the entanglement entropy problem, one wants to decide whether a certain bipartite state $\ket{\psi}_{AB}$ has low or small entanglement entropy between the subsystems $A$ and $B$. We consider the the ($2$-R\'enyi) entanglement entropy $S_2(\cdot)$, which for a mixed state $\rho_A$ is defined as
\begin{align*}
    S_2(\rho_A) = - \ln[\tr[\rho_A^2]].
\end{align*}

\begin{claim}[Lower bound for entanglement entropy] Let $U = \mathbb{I} - 2\ketbra{\psi}$ be a unitary for some bipartite quantum state $\ket{\psi} = \ket{\psi}_{AB} \in \mathbb{C}^d \otimes \mathbb{C}^d $.  Suppose that either (i) $S_2 (\tr_B(\ketbra{\psi}) ) \leq a$ or (ii) $S_2 (\tr_B(\ketbra{\psi}) ) \geq b$ with $b-a \geq \Delta$. Then any $\mathsf{C}$-tester, where $\mathsf{C} \subseteq \QMA(2)\slash\mathsf{qpoly}$, that decides whether (i) or (ii) holds with probability $\geq 2/3$ has to make at least
\begin{align*}
    T \geq \Omega\left(\frac{1}{\sqrt{\Delta}}\right)
\end{align*}
(controlled) queries to $U$ (or its inverse).
\label{clm:QEP}
\end{claim}

\begin{proof}
We reduce from the following unitary channel discrimination problem: $U_i = \mathbb{I} - 2\ketbra{\psi_i}$, $i \in [1,2]$, with 
\begin{align}
    \ket{\psi_1} = \frac{1}{\sqrt{2}}(\ket{11} +\ket{00)}, \quad \ket{\psi_2} = \sqrt{\frac{1+\sqrt{\Delta}}{2}} \ket{00} + \sqrt{\frac{1-\sqrt{\Delta}}{2}}\ket{11}
\end{align}
We have 
\begin{align*}
    \rho_{1,A} = \tr_B [\ket{\psi_1}\bra{\psi_1}] = I_2
\end{align*}
which has $S_2 (\rho_{1,A}) = \ln(2)$ and  
\begin{align*}
    \rho_{2,A} = \tr_B [\ket{\psi_2}\bra{\psi_2}] = \frac{1+\sqrt{\Delta}}{2}\ket{0}\bra{0} + \frac{1-\sqrt{\Delta}}{2} \ket{1}\bra{1}
\end{align*} which has $S_2 (\rho_{2,A}) = -\ln((1 + \Delta)/2) \leq \ln(2)-\Delta/2$, so $b-a \geq  \Delta/2$. Hence, if we could compute $S_2(\rho_i)$ up to precision $< \Delta/4$ for $i\in \{1,2\}$ given access to $U_i$ we could distinguish $U_1$ from $U_2$.
We have that
    \begin{align*}
        \eig(U_1^\dagger U_2) = \{1,-i \sqrt{\Delta} + \sqrt{1 - \Delta}, i \sqrt{\Delta} + \sqrt{1 - \Delta}\},
    \end{align*}
    which means that $0 \notin \conv(\eig(U_1^\dagger U_2)$  for $\Delta <1$.  We can brute force over all combinations of eigenvalues to find $
        D = \sqrt{1-\Delta} $
    and thus
     \begin{align*}
    \frac{1}{2}\dnorm{\mathcal{U}(U_1)-\mathcal{U}(U_2)} = \sqrt{1-(1-\Delta)} = \sqrt{\Delta}
\end{align*}
for $\Delta >0$. By~\cref{thm:lb_technique}, we find a lower bound of
$
    T \geq \Omega(1/\sqrt{\Delta})
$.
\end{proof}
This removes the logarithmic factor of $\tilde{\Omega}(1/\sqrt{\Delta})$ in Ref.~\cite{wang2023quantum}, and also resolves their open question whether their bound could be made to hold for a $\QMA$-tester (we show in fact that it is robust against even stronger classes of testers). From the proof of~\cref{clm:QEP} it is easy to show that the following corollary holds:

\begin{corollary}The same lower bound applies when given copies of $\ket{\psi}$ instead of the unitary $U = \mathbb{I} - 2\ketbra{\psi}$ from~\cref{clm:QEP}.
\label{cor:EE_copies}
\end{corollary}
\begin{proof}
We consider the same states as in the proof of~\cref{clm:QEP}. We can write $\ket{00}$ in terms of the two-dimensional subspace $\{\ket{\psi_2},\ket{\psi_2^\perp}\}$ as 
\begin{align*}
    \ket{00} = \sqrt{\frac{1+\sqrt{\Delta}}{2}}\ket{\psi_2} + \sqrt{\frac{1-\sqrt{\Delta}}{2}} \ket{\psi_2^\perp}.
\end{align*}
Consider the circuit which applies $U_2 = \mathbb{I} -2\ket{\psi_2}\bra{\psi_2}$ to $\ket{00}$ controlled by a qubit in $\ket{+}$ and then measures this first qubit in the Hadamard basis. We have
\begin{align*}
    cU_2 \ket{+}\ket{00} &= \frac{1}{\sqrt{2}}
    \ket{0}\left(\sqrt{\frac{1+\sqrt{\Delta}}{2}}\ket{\psi_2} + \sqrt{\frac{1-\sqrt{\Delta}}{2}} \ket{\psi_2^\perp}\right) +  \\
    & \qquad \frac{1}{\sqrt{2}} \ket{1} \left(-\sqrt{\frac{1+\sqrt{\Delta}}{2}}\ket{\psi_2} + \sqrt{\frac{1-\sqrt{\Delta}}{2}} \ket{\psi_2^\perp}\right) \\
        &= \sqrt{\frac{1+\sqrt{\Delta}}{2}}\ket{-}\ket{\psi_2} + \sqrt{\frac{1-\sqrt{\Delta}}{2}}\ket{+}\ket{\psi_2^\perp} 
    \end{align*}
which yields the state $\ket{-}\ket{\psi_2}$ when the measurement outcome is $\ket{-}$, which happens with probability $(1+\sqrt{\Delta})/2 \geq \frac{1}{2}$. For $\ket{\psi_1}$ the same bound holds since this is the case where $\Delta = 0$. Therefore, with an expected number of $\mO(T')$ runs of the circuit $T'$ copies can be created. Hence, by Markov's inequality, we have that there exists a quantum algorithm that creates $T'$ copies of $\ket{\psi_i}$ with probability $1-\delta$, making only $\mO(T'/\delta)$ queries to $U_i$. Hence, if we could solve the problem with high success probability using $o(T')$ copies, this would contradict the lower bound of~\cref{clm:QEP}.  
\end{proof}

A standard quantum in the literature to estimate the entanglement entropy, usually considered in the setting where the input to the state is given in a sample setting, is through the use of the SWAP-test~\cite{foulds2021controlled}. This gives an upper bound on the query complexity in terms of the precision $\Delta$, the dimension $d$, and an upper bound on the entanglement entropy $S_{\textup{upper}}$.

\begin{proposition}
Let $d = 2^n$ for some $n \in \mathbb{N}$, and $b,a \in [0, \ln(d)]$ with $b - a = \Delta > 0$. Suppose that we are given a number $S_{\textup{upper}} \in [0, \log d]$, and access to a bipartite state $\ket{\psi} \in \mathbb{C}^d \otimes \mathbb{C}^d$ via a unitary $U = \mathbb{I} - 2\ketbra{\psi}$, for which $S_2 (\ket{\psi}) \leq S_{\textup{upper}}$ holds. Then there exists a quantum algorithm that solves the entanglement entropy problem using 
\begin{align*}
  \mO\left(\frac{\sqrt{d} e^{S_\textup{upper}}}{\Delta}\right)
\end{align*}
queries to $U$, with probability $\geq 2/3$.
\end{proposition}

\begin{proof} 
    Let $\mathcal{A}$ be the quantum circuit which creates the maximally entangled state from the all-zeros state, i.e.
    \begin{align*}
        \mathcal{A} \ket{0\dots 0} = \frac{1}{\sqrt{d}} \sum_{j = 0}^{d-1} \ket{j} \ket{j} = \frac{1}{\sqrt{d}} \sum_{j = 0}^{d-1} \ket{\psi_j} \ket{\bar{\psi}_j},
    \end{align*}
    using the fact that the maximally entangled state can be written in any orthonormal basis. By using exact amplitude amplification, we can prepare $\ket{\psi}$ exactly with probability 1 using $\mO(\sqrt{d})$ queries to $U$ and $\mathcal{A}$~\cite{hoyer2000arbitrary}. Let $\mathcal{B}$ be the algorithm which does this. If the SWAP-test~\cite{buhrman2001quantum} is applied to two copies of a mixed state $\rho$, the probability of measuring $0$ on the first qubit is given by $\frac{1}{2} + \frac{1}{2}\tr[\rho^2]$ (a proof of this can be found in Ref.~\cite{kobayashi2003quantum}). Using the SWAP test in conjunction with quantum amplitude estimation~\cite{brassard2002quantum}, we can estimate $\tr[\rho_A^2]$ up to additive precision $\epsilon$ using $\mO(1/\epsilon)$ copies of $\ket{\psi}$ (i.e., calls to $\mathcal{B}$) with high probability. To make the estimation, when it succeeds, biased towards always returning an overestimate, we can apply the simple trick of providing an estimate $\bar{x}$ up to $\epsilon/2$ additive precision and making our new estimate $\hat{x} := \bar{x} + \epsilon/2$. Let $x := \tr[\rho_A^2]$. If we are guaranteed an upper bound on the entanglement entropy of $S_\textup{upper}$, then we have that $x \in [e^{-S_\textup{upper}}, 1]$. Choosing $\epsilon = \Delta e^{- S_\textup{upper}}$, we have
   \begin{align*}
   \abs{-\ln x - (-\ln \hat{x})} &= \abs{ -\ln (x + \epsilon) - (- \ln (x)) } \\
       &= \abs{ \ln\left(\frac{x}{x + \epsilon} \right) } \\
       &= \ln (1 + \epsilon/x) \\
       &\leq \ln (1 + \epsilon e^{S_{\textup{upper}}}) \\
       &\leq \epsilon e^{S_{\textup{upper}}} \\
       &\leq \Delta/4,
   \end{align*}
   which is a sufficient precision to distinguish between the two cases. Hence, we need to make a total of
   \begin{align*}
       T = \mO\left(\sqrt{d} \cdot \frac{1}{\epsilon}\right) = \mO\left(\frac{\sqrt{d} e^{S_\textup{upper}}}{\Delta}\right)
   \end{align*}
 queries to $U$.
\end{proof}

\noindent Our lower bound does not incorporate the dependence on the dimension, nor does it account for the fact that the estimation might become harder as the entanglement entropy approaches its maximum value. The lack of dimension-dependence might be explained by the fact that our lower bound does not take into account the difficulty of preparing $\ket{\psi}$ given access to $U_i$, as indicated by~\cref{cor:EE_copies}.

However, even when $S_\textup{upper}$ and $d$ are constant, this bound is still quadratically worse in the precision $\Delta$ than the lower bounds from~\cref{clm:QEP} and Ref.~\cite{wang2023quantum}. We leave it as an open question whether these bounds (upper or lower) can be improved, especially in terms of achieving a better dependency on the precision $\Delta$.

\subsubsection{Subset support verification and amplitude estimation}
Here we consider a variant to the amplitude estimation problem, which we will call the \textit{subset support verification} problem. In this problem, one is given access to a unitary $U$ which prepares a subset state for some subset $S \subseteq \{0,1\}^n$ when applied to the all-zeros state, as well as a bit string $j \in \{0,1\}^n$. The task is to decide whether $j \in S$ or $j \notin S$.

\begin{claim}[Lower bound for subset support verification] Let $U$ be a unitary that applies a transformation of the form
\begin{align*}
    U \ket{0^n}  = \frac{1}{\sqrt{|S|}}\sum_{i \in S} \ket{i},
\end{align*}
and let $j \in \{0,1\}^n$. Then any $\mathsf{C}$-tester, where $\mathsf{C} \subseteq \QMA(2)\slash\mathsf{qpoly}$, that decides whether $j\in S$ or $j \notin S$ with probability $\geq 2/3$ has to make at least
\begin{align*}
    T \geq \Omega\left({\sqrt{|S|}}\right)
\end{align*}
(controlled) queries to $U$ (or its inverse).
\label{clm:SSV}
\end{claim}
\begin{proof}
Let $S \subseteq \{0,1\}^n$ be any subset that contains $0^n$, and let $j = 0^n$ (the proof can easily be modified to work for any choice of $j$). We can construct an explicit unitary $U$ by setting $U = 2\ketbra{v}-\mathbb{I}$, where
\begin{align*}
    \left(2\ketbra{v}-\mathbb{I}\right)\ket{0^n} = \ket{S}.
\end{align*}
Solving for $\ket{v}$, using that $\bra{0 \dots 0} \ket{S}$ is real we obtain 
\begin{align*}
    \ket{v} = \frac{\ket{0^n} + \ket{S}}{\sqrt{2(1+\bra{0 \dots 0}\ket{S})}}.
\end{align*}
Let $S_1 = S$ and $S_2 = S_1 \setminus \{0^n\}$. We define
\begin{align*}
    \ket{v_1} = \frac{\ket{0^n} + \ket{S_1}}{\sqrt{2(1+\sqrt{1/|S|})}},\quad \ket{v_2} = \frac{\ket{0^n} + \ket{S_2}}{\sqrt{2}},
\end{align*}
and
\begin{align*}
    U_1 = 2\ketbra{v_1}-\mathbb{I},\quad U_2 = 2\ketbra{v_2}-\mathbb{I}.
\end{align*}
We can write $\ket{v_2} = \sqrt{1-\alpha} \ket{v_1} + \sqrt{\alpha}  | v_1^\perp \rangle$ for some $\alpha \in [0,1]$. We have
\begin{align*}
    U_2 = 2 \left[(1-\alpha)\ketbra{v_1} + \sqrt{1-\alpha}\sqrt{\alpha} \ket{v_1} \langle v_1^\perp| + \sqrt{1-\alpha}\sqrt{\alpha} \ket{v_1} \!\langle v_1^\perp | + \alpha | v_1^\perp \rangle \! \langle v_1^\perp |  \right]   -\mathbb{I}.
\end{align*}
Writing $U_1^\dagger U_2 $ in the $\{ \ket{v_1},\ket{v_1^\perp}, \dots\}$ basis we obtain
\begin{align*}
    U_1^\dagger U_2 = \begin{bmatrix} B  & 0\\
    0 & I_{2^{n-1}}
    \end{bmatrix},
\end{align*}
where $B$ is a $2 \cross 2$-matrix given by
\begin{align*}
    B = \begin{bmatrix} 1-2 \alpha  & 2 \sqrt{1-\alpha}\sqrt{\alpha}\\
    -2 \sqrt{1-\alpha}\sqrt{\alpha} &1-2 \alpha
    \end{bmatrix}.
\end{align*}
Since $U_1^\dagger U_2$ is a block-diagonal matrix, its eigenvalues are given by
\begin{align*}
    \eig(U_1^\dagger U_2) = \eig(B)  \cup \eig(\mathbb{I}_{2^{n-1}})  = \{1, 1 - 2 \alpha -2 \sqrt{\alpha^2-\alpha}, 1 - 2 \alpha +2 \sqrt{\alpha^2-\alpha} \}.
\end{align*}
Again, it is easy to see that for $\alpha<1$ we have
 $0 \notin \conv(\eig(U_1^\dagger U_2)$.
Therefore, $D = 1-2\alpha$ and
\begin{align*}
    \frac{1}{2}\dnorm{(\mathcal{U}(U_1) - \mathcal{U}(U_2))} = 2 \sqrt{\alpha-\alpha^2}  \leq 2 \sqrt{\alpha} \leq 2\sqrt{\frac{1}{|S|}}
\end{align*}
since,
\begin{align*}
    \sqrt{\alpha} &= \sqrt{1-\abs{\bra{v_1}\ket{v_2}}^2}\\
    &=  \sqrt{1- \abs{ \frac{1 + \bra{0^n} \ket{S_2} + \bra{S_1}\ket{0^n}+ \bra{S_1}\ket{S_2}}{2 \sqrt{1+\sqrt{1/|S|}}} }^2}\\
    &=  \sqrt{1- \abs{ \frac{1 + \sqrt{1/|S|}+  \sqrt{(|S|-1)/|S|}   }{2 \sqrt{1+\sqrt{1/|S|}}} }^2}\\
    &=\sqrt{\frac{1}{2}-\frac{1}{2}\sqrt{\frac{|S|-1}{|S|}}}\\
    &\leq \sqrt{\frac{1}{2}-\frac{1}{2}{\frac{|S|-1}{|S|}}}\\
    &\leq  \sqrt{\frac{1}{|S|}}.
\end{align*}
Hence, by~\cref{thm:lb_technique} we have that
$
    T \geq \mO\left(\sqrt{|S|}\right)
$
to distinguish $U_1$ and $U_2$ with probability $\geq 2/3$.
\end{proof}

The upper bound of $\mO(\sqrt{S})$ follows directly from quantum amplitude estimation~\cite{brassard2002quantum}, for which the optimality is also a corollary of~\cref{clm:SSV}, as it proves a lower bound in a more restricted setting.

\begin{corollary}[Quantum amplitude estimation lower bound] Given a unitary $U$ which acts as $U \ket{0^n} = \sqrt{\alpha} \ket{\psi} +\sqrt{1-\alpha} |\psi^\perp\rangle$. Suppose that either (i) $\sqrt{\alpha} \leq a$ or (ii) $\sqrt{\alpha}  \geq b$ with $b-a = \epsilon$. Then any $\mathsf{C}$-tester, where $\mathsf{C} \subseteq \QMA(2)\slash\mathsf{qpoly}$, that decides whether (i) or (ii) holds with probability $\geq 2/3$ has to make at least
\begin{align*}
    T \geq \Omega\left(\frac{1}{\epsilon}\right)
\end{align*}
(controlled) queries to $U$ (or its inverse).
\label{cor:QAE}
\end{corollary}
\begin{proof}
    Let $\ket{S} = \frac{1}{\sqrt{|S|}}\sum_{i \in S} \ket{i}$ with $|S| = 1/\epsilon^2$, where $\epsilon$ is chosen such that $|S|$ is integer (this assumption does not change the bound qualitatively). Let $\ket{\psi} =\ket{0^n}$. Note that
\begin{align*}
    \sqrt{\alpha_1} = \bra{0^n} U_1 \ket{0^n} = 1/\sqrt{|S|}, \quad \sqrt{\alpha_2} = \bra{0^n}  U_2 \ket{0^n} = 0, 
\end{align*}
so deciding whether $\sqrt{\alpha_i}=\frac{1}{ \sqrt{|S|}} = \epsilon$ or  $\sqrt{\alpha_i} = 0$ is sufficient to distinguish $U_1$ from $U_2$. By the proof of~\cref{clm:SSV}, we then find $
    T \geq \Omega(1/\epsilon)
$,
completing the proof.
\end{proof}

\subsection{Other unitary problems}
The next two examples are not unitary property testing cases, but quantum algorithmic primitives involving the implementation of some unitary given access to some building block in the form of a \textit{block-encoding}~\cite{low2017optimal,low2019hamiltonian}.
\begin{definition}[Block-encoding] Let $M$ be $2^n \cross 2^n$-dimensional matrix, $\alpha,\epsilon > 0$ and $a \in \mathbb{N}$. An $(n+a)$-qubit unitary operator $U$ is an $(\alpha,a,\epsilon)$-block-encoding of $M$ if
\begin{align*}
    \opnorm{\alpha (\bra{0}^{\otimes a} \otimes \mathbb{I}) U (\ket{0}^{\otimes a}  \otimes \mathbb{I}) -M} \leq \epsilon.
\end{align*}
If $\alpha = a =1$ and $\epsilon =0$,  we simply say that $U$ is a block-encoding of $M$.
\label{def:block_encoding}
\end{definition}
\noindent Given some $n$-qubit operator $M$, one can construct a $n+1$-qubit unitary operator $U$ provided all singular values of $M$ are upper bounded by 1 in the following way. Let $M=R\Sigma V^\dagger$ be the singular value decomposition (SVD) of $M$. Then 
\begin{align}
    U = \begin{pmatrix} M & R \sqrt{\mathbb{I}-\Sigma^2}V^\dagger \\ R \sqrt{\mathbb{I}-\Sigma^2}V^\dagger & -M
    \end{pmatrix}
    \label{eq:block_constr}
\end{align}
is a $(1,1,0)$-block-encoding of $M$, since
\begin{align*}
    (\bra{0}\otimes I_{2}) U (\ket{0} \otimes I_{2}) =M.
\end{align*}

\subsubsection{Thermal state preparation (quantum Gibbs sampling)}
For some Hamiltonian $H$, the thermal state $\rho_{\beta}$ at inverse temperature $\beta$ is defined as 
\begin{align}
    \rho_{\beta} = \frac{e^{- \beta H}}{\tr[e^{- \beta H}]}.
    \label{eq:TS}
\end{align}
We say a quantum algorithm is an \textit{approximate Gibbs sampler} if it prepares the thermal state $\rho_{\beta}$ up to some trace distance $\epsilon$. In Ref.~\cite{kastoryano2023quantum}, an optimal lower bound in $\beta$ is proven in Appendix G. We demonstrate that this bound can also be derived using our framework. The main distinction is that we prove the result by directly examining the diamond distance between the block-encodings of the Hamiltonians, rather than using reflections about a purified Gibbs state.

\begin{claim}[Lower bound for quantum Gibbs sampling] Let $H$ be some Hamiltonian to which we have access through a block-encoding $U_H$. Suppose $\beta \geq \sqrt{\frac{14}{3}} \approx 2.16$. Then it takes at least 
\begin{align*}
T \geq \Omega(\beta)    
\end{align*}
queries to $U_H$ to prepare the thermal state at inverse temperature $\beta$ up to trace distance $\leq \frac{1}{24}$.
\label{claim:QGS}
\end{claim}
\begin{proof}
    Let $H_1= \left(\frac{1}{2}+\frac{1}{\beta}\right) \ket{0}\bra{0} + \left(\frac{1}{2}-\frac{1}{\beta}\right)  \ket{1}\bra{1}$, and $H_2= \left(\frac{1}{2}-\frac{1}{\beta}\right) \ket{0}\bra{0} + \left(\frac{1}{2}+\frac{1}{\beta}\right)  \ket{1}\bra{1}$. We have that by~\cref{eq:TS} the thermal states are given by
    \begin{align*}
        \rho_{1,\beta} = \frac{1}{1+e^2}\begin{bmatrix}
            1 & 0\\
            0 & e^2
        \end{bmatrix}
    \end{align*} and
        \begin{align*}
        \rho_{2,\beta}= \frac{1}{1+e^2}\begin{bmatrix}
            e^2 & 0\\
           0 & 1.
        \end{bmatrix}
    \end{align*}
Therefore,
\begin{align*}
    \frac{1}{2}\tnorm{\rho_{1,\beta}-\rho_{2,\beta}} = 1-\frac{2}{1+e^2} \geq \frac{3}{4}.
\end{align*}
Now suppose that we can only prepare some $\tilde{\rho}_{i,\beta}$ such that $\frac{1}{2}\tnorm{\tilde{\rho}_{i,\beta}-{\rho}_{i,\beta}} \leq \epsilon$ for $i \in \{1,2\}$.  By applying the reverse triangle inequality twice, we find that
\begin{align*}
    \frac{1}{2}\tnorm{\tilde{\rho}_{1,\beta}-\tilde{\rho}_{2,\beta}} & = \frac{1}{2}\tnorm{\rho_{1,\beta} - \rho_{2,\beta} - (\rho_{1,\beta}-\tilde{\rho}_{1,\beta}) -(\tilde{\rho}_{2,\beta} - \rho_{2,\beta}) }\\
    &\geq \frac{1}{2}\abs{\tnorm{\rho_{1,\beta} - \rho_{2,\beta} - (\rho_{1,\beta}-\tilde{\rho}_{1,\beta})} - \tnorm{\tilde{\rho}_{2,\beta} - \rho_{2,\beta}} }\\
    &\geq \frac{1}{2}\abs{ \abs{\tnorm{\rho_{1,\beta} - \rho_{2,\beta}} -\tnorm{\rho_{1,\beta}-\tilde{\rho}_{1,\beta}}} - \tnorm{\tilde{\rho}_{2,\beta} - \rho_{2,\beta}} }\\
    &\geq \frac{3}{4}- 2 \epsilon \geq \frac{2}{3}
\end{align*}
when $\epsilon \leq \frac{1}{24}$, which means that $\tilde{\rho}_{1,\beta}$ and $\tilde{\rho}_{2,\beta}$ can be distinguished with success probability $\geq 2/3$. Hence, if $\tilde{\rho}_{i,\beta}$ can be constructed using the block-encoding $U_i$ of $H_i$, we have a distinguisher for unitary channels associated with $U_i$ for $i \in \{1,2\}$.

The SVDs of $H_1$ and $H_2$ are $\mathbb{I} H_1 \mathbb{I}$ and $\mathbb{I} H_2 \mathbb{I}$, since both are diagonal and PSD. Using~\cref{eq:block_constr} we can construct the following two block-encodings of $H_1$ and $H_2$:
\begin{align*}
        U_1= \begin{bmatrix}
   \frac{1}{2}+\frac{1}{\beta} & 0 & \sqrt{1-
   \left(\frac{1}{2}+\frac{1}{\beta}\right)^2} & 0\\
   0 & \frac{1}{2}-\frac{1}{\beta} & 0 & \sqrt{1-
   \left(\frac{1}{2}-\frac{1}{\beta}\right)^2}\\
   \sqrt{1-
   \left(\frac{1}{2}+\frac{1}{\beta}\right)^2} & 0 & -\frac{1}{\beta}-\frac{1}{2} & 0\\
   0 & \sqrt{1-
   \left(\frac{1}{2}-\frac{1}{\beta}\right)^2} & 0 & \frac{1}{\beta}-\frac{1}{2}
        \end{bmatrix},
\end{align*}
\begin{align*}
        U_2= \begin{bmatrix}
   \frac{1}{2}-\frac{1}{\beta} & 0 & \sqrt{1-
   \left(\frac{1}{2}-\frac{1}{\beta}\right)^2} & 0\\
   0 & \frac{1}{2}+\frac{1}{\beta} & 0 & \sqrt{1-
   \left(\frac{1}{2}+\frac{1}{\beta}\right)^2}\\
   \sqrt{1-
   \left(\frac{1}{2}-\frac{1}{\beta}\right)^2} & 0 & -\frac{1}{\beta}+\frac{1}{2} & 0\\
   0 & \sqrt{1-
   \left(\frac{1}{2}+\frac{1}{\beta}\right)^2} & 0 & -\frac{1}{\beta}-\frac{1}{2}
        \end{bmatrix}.
\end{align*}
Evaluating the operator norm distance of $U_1$ and $U_2$ gives us
\begin{align*}
    \opnorm{U_1-U_2} = \frac{\sqrt{3 \beta^2-\sqrt{9 \beta^4-40 \beta^2+16}+4}}{\sqrt{2} \beta} \leq \frac{}{} \frac{3}{\beta}
\end{align*}
for $\beta \geq \sqrt{\frac{14}{3}}$. Moreover, for these values of $\beta$ we have $\opnorm{U_1-U_2} \leq  \frac{3}{\sqrt{14/3}} \leq 1.4 $, which means that $0 \notin \conv(\eig(U_1^\dagger U^2)$ by~\cref{lem:op_distance_convhull}.
Therefore, by~\cref{lem:norm_conv2} we then have
$
    T \geq \Omega(\beta)
$.
\end{proof}
\noindent This matches the upper bound of the known quantum Gibbs samplers, see for example Ref.~\cite{chen2023efficient}.

\subsubsection{No fast-forwarding for Hamiltonian simulation}
In Hamiltonian simulation one has access to some Hamiltonian $H$ and is given a time $t \in \mathbb{R}$, with the goal of implementing the unitary $\tilde{U}$ which approximates $U = e^{-i t H}$ up to diamond distance $\epsilon$. It is well-known that it is generally not possible to do so-called Hamiltonian \textit{fast-forwarding}, which refers to a Hamiltonian simulation algorithm which implements the $\tilde{U}$ in time sub-linear in $t$~\cite{berry2007efficient}.   

\begin{claim}[No fast-forwarding for Hamiltonian simulation] Let $H$ with $\opnorm{H} \leq 1$ be some Hamiltonian to which we have access through a block-encoding $U_H$. Suppose $t \geq 1/2\pi$. Then it takes at least 
\begin{align*}
T \geq \Omega(t)    
\end{align*}
queries to $U_H$ to implement $e^{-i H t}$ up to diamond distance $\leq 1/3$.
\label{claim:HS}
\end{claim}
\begin{proof}
    Let $H_1= \frac{I_2}{2}$, and $H_2= H_1 + \frac{1}{2t} \ket{1}\bra{1}$, such that $\opnorm{H_1} \leq 1$ and $\opnorm{H_2} \leq 1$. Define $t = 2\pi t'$ with $t' \geq 0$. Suppose we could implement $U_i = e^{-i H_i t}$ perfectly for $i \in \{1,2\}$. We then have that
    \begin{align*}
        e^{- 2 \pi i H_1 t'}\ket{+} = e^{- i \pi t'} \ket{+}
    \end{align*}
    and
    \begin{align*}
        e^{-2 \pi i H_2 t'}\ket{+} = e^{- i \pi t'} \ket{-},
    \end{align*}
    which are perfectly distinguishable by a measurement in the Hadamard basis. If instead we can implement some $\tilde{U}_i$ such that $\frac{1}{2}\dnorm{\tilde{U}_i - U_i} \leq 1/3$ for $i \in \{1,2\}$, our Hadamard basis measurement will be able to distinguish both cases with success probability $\geq 2/3$. Hence, we have that a Hamiltonian simulation algorithm that uses $U_H$ as a subroutine can be used to distinguish $U_{H_1}$ from $U_{H_2}$. For the block-encodings of $H_1$ and $H_2$ we have 
    \begin{align*}
        U_1= \begin{bmatrix}
   \frac{1}{2} & 0 & \frac{\sqrt{3}}{2} & 0\\
   0 & \frac{1}{2} & 0 & \frac{\sqrt{3}}{2}\\
   \frac{\sqrt{3}}{2} & 0 & -\frac{1}{2} & 0\\
   0 & \frac{\sqrt{3}}{2} & 0 &-\frac{1}{2}
        \end{bmatrix}, \quad 
        U_2= \begin{bmatrix}
   \frac{1}{2} & 0 & \frac{\sqrt{3}}{2} & 0\\
   0 & \frac{1+1/t'}{2} & 0 & \sqrt{1+\left(\frac{1}{2}+\frac{1}{2t'^2}\right)^2}\\
   \frac{\sqrt{3}}{2} & 0 & -\frac{1}{2} & 0\\
   0 & \sqrt{1+\left(\frac{1}{2}+\frac{1}{2t'^2}\right)^2} & 0 &-\frac{1+1/t'}{2}
        \end{bmatrix},
\end{align*}
respectively. Again by~\cref{lem:norm_conv2}, evaluating the operator distance we find
\begin{align*}
     \opnorm{U_1-U_2} = \frac{\sqrt{3-\sqrt{-\frac{3}{t'^2}-\frac{6}{t'}+9}-\frac{1}{t'}}}{\sqrt{2}} \leq \frac{1}{t'},
\end{align*}
for which the inequality can be shown to hold after some straightforward algebraic manipulation, assuming that $t' \geq 1$. Again, for $t' \geq 1$ (and thus $t \geq 1/2 \pi$) we have $\opnorm{U_1-U_2} \leq 1 \leq 1.4$, so $0 \notin \conv(\eig(U_1^\dagger U^2)$ by~\cref{lem:op_distance_convhull}. Hence, by~\cref{thm:lb_technique}, $
    T = \Omega(t)$.
\end{proof}

See Ref.~\cite{gilyen2019quantum} for a matching upper bound in the simulation time $t$.

\subsubsection{Hamiltonian learning in the Heisenberg limit}
In Hamiltonian learning, the problem is to output a classical description of a Hamiltonian $H'$ that is $\epsilon$-close to $H$ with respect to some distance measure $d$ in the space of Hamiltonians, with probability $\geq$ by making queries to its time evolution operator $U(t) = e^{-i t H}$, where $t \geq 0$ is tunable. The total evolution time is then defined as the sum of all different evaluation times used in the different queries. As a distance measure, we will take the operator norm distance $\opnorm{\tilde{H}-H} $.

\begin{claim} Given access to an unknown Hamiltonian $H$ with $\opnorm{H} \leq 1$ through its time evolution operator $U(t) = e^{-i t H}$, with $t \geq 0$ a tunable parameter, outputting the matrix description of a $\tilde{H}$ such that $\opnorm{\tilde{H}-H} \leq \epsilon$, where $0 < \epsilon \leq \frac{1}{2}$, requires a total evolution time of at least $\Omega(1/\epsilon)$.
\label{claim:learning}
\end{claim}
\begin{proof}
Let $H_1 = \mathbb{I}$ and $H_2 = \mathbb{I}-2\epsilon \ketbra{0}$. Since $0 < \epsilon \leq \frac{1}{2}$, we have $\opnorm{H_1} \leq 1$ and $\opnorm{H_2} \leq 1$ holds. Learning an unknown $U$ up to $\epsilon$ in operator distance allows one to distinguish $H_1$ from $H_2$. Suppose that we make $m$ queries to the time evolution operator, each with some evolution time $t_j$ with $j \in [m]$. We have that $U_1(t) = e^{-i t_j} \mathbb{I}$ and 
\begin{align*}
    U_2(t_j) = \begin{bmatrix}
        e^{-i t_j(1-2\epsilon)} & 0 \\ 0 & e^{-i t_j}
    \end{bmatrix}.
\end{align*}
We have that
\begin{align*}
    U_1(t_j)^\dagger U_2(t_j) = \begin{bmatrix}
        e^{2i t_j \epsilon} & 0 \\ 0 & 1
    \end{bmatrix}.
\end{align*}
which means that $D = \frac{1}{2} \abs{1+e^{2 i t_j \epsilon}}$. Now suppose that $t_j \geq \frac{\pi}{2 \epsilon}$. Then this would match our desired lower bound. However, if $t_j \epsilon < \frac{\pi}{2 \epsilon}$, the condition $0 \notin \conv(\eig(U_1^\dagger U_2))$ is satisfied. Therefore, by~\cref{lem:norm_conv1} we have
\begin{align*}
    \frac{1}{2}\dnorm{\mathcal{U}(U_1(t_j)) - \mathcal{U}(U_2(t_j))} = \sqrt{1-D^2} = \abs{\sin(t_j\epsilon)} \leq t_j\epsilon.
\end{align*}
The proof of~\cref{thm:lb_technique} then directly implies that $T = \sum_{j \in [m]} t_j = \Omega(1/\epsilon)$.
\end{proof}
See Ref.~\cite{huang2023learning} for a (local) Hamiltonian learning algorithm which achieves the Heisenberg scaling for its total evolution time.

\subsubsection{Ground state preparation of spectrally-gapped Hamiltonians}
In approximate ground state preparation problems one is given a Hamiltonian $H$ with some unknown ground state $\ket{\psi}$, for which one has to prepare a state $\ket{\psi}$ such that $\abs{\bra{\psi_0} \ket{\psi}}^2\geq 1-\epsilon$. It is well known that the \textit{spectral gap} $\gamma(H)$, that is the difference between the energies of the ground state and the first excited state of $H$, plays an important role in how difficult this problem is~\cite{Lin2020nearoptimalground,arad2017rigorous,Deshpande2020}. 

\begin{claim}[Spectrally gapped ground state preparation] Let $H$ with $\opnorm{H} \leq 1$ be some Hamiltonian to which we have access through a block-encoding $U_H$. Suppose $\gamma(H) \geq \Delta$ for some $0 < \Delta \leq 1$. Then it takes at least 
\begin{align*}
T \geq \Omega(1/\Delta)    
\end{align*}
queries to $U_H$ to implement a state $\ket{\phi}$ which approximates the ground state $\ket{\psi}$ up to fidelity $\geq 2/3$.
\label{claim:GSP}
\end{claim}
\begin{proof} For a Hilbert space $\mathcal{H} := \mathbb{C}^3$, let $H_1 = \Delta \ketbra{1} +\ketbra{2}$ and $H_2 = \Delta \ketbra{0} +\ketbra{2}$. Since $0 < \Delta \leq 1$, we have that $\opnorm{H_1} \leq 1$ and $\opnorm{H_2} \leq 1$ holds. The ground states of $H_1$ and $H_2$ are given by $\ket{\psi_1} = \ket{0} $ and $\ket{\psi_2} = \ket{1} $, respectively. Note that indeed $\gamma(H_1) = \gamma(H_2)) = \Delta$.  Suppose that we have an algorithm could prepare the ground state of an unknown $H$ up to fidelity $\geq 2/3$ by making queries to the block-encoding of $H$, then we could distinguish $H_1$ from $H_2$ by preparing the ground state and performing a measurement $\{\ketbra{0},\mathbb{I}-\ketbra{0}\}$. For some real number $c = \abs{\bra{\phi} \ket{\psi}}$ with $|c|^2 \geq 2/3$, and $\theta \in [0,1]$, we can write the prepared state as $\ket{\phi} = ce^{i \theta }\ket{\psi} + (1-c) |\psi^\perp\rangle$, where $|\psi^\perp\rangle$ lives in the space orthogonal to $\ket{\psi}$.  Clearly, if $\ket{\psi}=\ket{0}$ then we have that $|\bra{\phi} \ket{0}| = |c|^2 \geq 2/3 $ and when $\ket{\psi}=\ket{1}$ then $|\bra{\phi} \ket{0}|^2 \leq 1- |c|^2 \leq 1/3$. Therefore, we have that the probability of measuring $\ket{0}$ when $H= H_1$ is $\geq 2/3$ and when $H = H_2$ it is $\leq 1/3$, which means that we can distinguish between both cases with an overall success probability $\geq 2/3$.
For the block-encodings of $H_1$ and $H_2$ we have
\begin{align*}
    U_1 = \begin{bmatrix}
0 & 0 & 0 & 1 & 0 & 0 \\
0 & \Delta & 0 & 0 & \sqrt{1 - \Delta^2} & 0 \\
0 & 0 & 1 & 0 & 0 & 0 \\
1 & 0 & 0 & 0 & 0 & 0 \\
0 & \sqrt{1 - \Delta^2} & 0 & 0 & -\Delta & 0 \\
0 & 0 & 0 & 0 & 0 & -1
\end{bmatrix}, \quad U_2 = \begin{bmatrix}
0 & 0 & 0 & \sqrt{1 - \Delta^2} & 0 & 0 \\
0 & \Delta & 0 & 0 & 1 & 0 \\
0 & 0 & 1 & 0 & 0 & 0 \\
\sqrt{1 - \Delta^2} & 0 & 0 & -\Delta & 0 & 0 \\
0 & 1 & 0 & 0 & 0 & 0 \\
0 & 0 & 0 & 0 & 0 & -1 \\
\end{bmatrix}.
\end{align*}
A direct computation shows that
\begin{align*}
    \opnorm{U_1 - U_2} = \frac{1}{2} \left( \Delta + \sqrt{8 - 3 \Delta^2 - 8 \sqrt{1 - \Delta^2}} \right) \leq 2 \Delta.
\end{align*}
For $\Delta \leq 0.7$ we therefore have that $\opnorm{U_1 - U_2} \leq 1/4$ and thus, by virtue of~\cref{lem:op_distance_convhull}, $0 \notin \conv(\eig(U_1^\dagger U_2))$. By~\cref{thm:lb_technique}, $
    T \geq \Omega(1/\Delta)$.
\end{proof}
For an algorithm that uses $H$ in its block-encoding and achieves $\tilde{\mO}(1/\Delta)$ scaling, see Ref.~\cite{Lin2020nearoptimalground}.

\section{Quantum oracle separations with \texorpdfstring{$\SBQP$}{SBQP}}
\label{sec:oracle_sep}
In this section we will give the proofs of~\cref{thm:orsepMP} and~\cref{thm:orsepadvice}. We will first argue that lower bounds on unitary property testing in general can be used to show quantum oracle separations using the following strategy:
Take a language $L$ which is not in $\mathsf{C}_1$ nor in $\mathsf{C}_2$ (if it already is in $\mathsf{C}_1$, then it already implies a separation). Define a unitary property $\mathcal{P} = (U_\text{yes},U_\text{no})$ in such a way that for a family of quantum oracles $\{U_x\}$, parametrised by the input $x$, the following two conditions hold:
\begin{enumerate}
    \item If $x \in L $, then $U_x \in U_\text{yes}$ and if $x \notin L$, then $U_x \in U_\text{no}$.
    \item To decide whether $U_x \in U_\text{yes}$ or $U_x \in U_\text{no}$ can be done in $\mathsf{C}_1$ (i.e. it does not exceed the maximum allowed query complexity) but not in $\mathsf{C}_2$ (it does exceed the maximum allowed query complexity for all $n \geq n_0$, for some $n_0 \in \mathbb{N}$).
\end{enumerate}
One can then use the standard technique of diagonalization~\cite{baker1975relativizations} to show a quantum oracle separation between $\mathsf{C}_1$ and $\mathsf{C}_2$. However, in our case, it turns out the separation can be shown in an even simpler way. Since we want oracle separations with $\SBQP$, we can construct an oracle that is very close in diamond distance to the identity operator. For any class that makes only a polynomial number of queries and can only distinguish between inverse polynomial output probabilities (and has error reduction), the unitary oracle will look just like the identity operator at large values of $n$. Hence, if a language $L$ is not in $\mathsf{C}_2$, it will also not be in $\mathsf{C}_2^U$ since this would imply that it would also be in $\mathsf{C}_2^\mathbb{I} = \mathsf{C}_2$. We will need the following lemma of one-bit quantum phase estimation to show containment in $\SBQP$.
\begin{lemma}[One-bit quantum phase estimation~\cite{kitaev1995quantum}] Let $U \in \mathbb{U}(2^n)$ and $\ket{\psi}$ a $n$-dimensional quantum state (which can be prepared exactly in an efficient manner) such that $U \ket{\psi} = e^{2\pi i \theta} \ket{\psi}$, for some $\theta \in [0,1)$. Then there exists a quantum algorithm that makes one controlled query to $U$ and outputs $1$ with probability $\sin^2(\pi \theta)$.
\label{lem:ob_QPE}
\end{lemma}
\begin{proof} The one-bit quantum phase estimation circuit for a $n$-qubit unitary starts with $n+1$ qubits initialised in $\ket{0} \ket{\psi}$, applies a Hadamard to the first qubit, followed by an application of $U$ to the remaining $n$ qubits controlled on the first qubit, and finishes by applying yet another Hadamard to first qubit followed by a measurement in the computational basis of this qubit. We have that the output state of the one-bit quantum phase estimation circuit is given by
\begin{align*}
    \frac{1}{2}(\ket{0}+\ket{1}) \ket{\psi} + \frac{e^{2 \pi i \theta}}{2}(\ket{0}-\ket{1}) \ket{\psi},
\end{align*}
which means that the probability of measuring `$1$' as a function of $\theta \in [0,1)$ can be expressed as
\begin{align*}
    \Pr[\text{One-bit QPE outputs $1$}] =  \abs{\frac{1-e^{2 i \pi \theta}}{2}}^2 = \sin^2(\pi \theta).
\end{align*}
\end{proof}
We are now ready to state the proof, which uses the above argument to turn the lower bound of~\cref{clm:QPE} into a quantum oracle separation.
\orsepMP*
\begin{proof}
Since $\QMA(2) \subseteq \NEXP$, we have that there must exist a unary language $L$ for which it holds that $L \notin \QMA(2)$. Pick any such $L$, let $p(n)$ be some polynomial and define the quantum oracle $U = \{U_n \}_{n \geq 1}$ with $U_n = e^{2 \pi i \theta_n} \ket{0}\bra{0} + \ket{1}\bra{1}$ for some $\theta_n \in [0,1)$ as follows:
\begin{itemize}
    \item if $0^n \in L$, then $\theta_n = 2^{-p(n)}$.
    \item if $0^n \notin L$, then $\theta_n = 0$.
\end{itemize}
Clearly, we have that $L \in \SBQP^U$ since we have that running the one-bit quantum phase estimation protocol from~\cref{lem:ob_QPE} with $U$ and $\ket{\psi} = \ket{0}$,  saying that it accepts when it outputs `$1$', satisfies
\begin{itemize}
    \item if $0^n \in L$, then $\Pr[\text{One-bit QPE accepts}] = \sin^2(\pi 2^{- p(n)}) \geq 2^{-2p(n)} \geq  2^{-p'(n)} $,
    \item if $0^n \notin L$, then $\Pr[\text{One-bit QPE accepts}] = 0 \leq 2^{-p'(n)-1}  $,
\end{itemize}
for some polynomial $p'(n)$. We will now show that  $L \notin \QMA(2)^U$. Let $q(n)$ be a polynomial which is an upper bound on the runtime of the $\QMA(2)^{U}$ verifier. Now observe that the identity operator has diamond distance $\leq \pi 2^{-p'(n)}$ from the oracle $U_n$. Hence, if we replace the oracle $U$ with the identity operator $\mathbb{I}$ in all the $\QMA(2)^{U}$ verifiers $V^U$, we have that by~\cref{thm:lb_technique} that
\begin{align*}
    \abs{\Pr[V^{U_n} \text{ accepts } (x,\ket{\psi_1} \otimes \ket{\psi_2})]-\Pr[V^{\mathbb{I}} \text{ accepts } (x,\ket{\psi_1} \otimes \ket{\psi_2})]} \leq \frac{\pi q(n)}{2^{p'(n)}} \leq 0.01,
\end{align*}
for all $n \geq n_0$, with $n_0$ some constant that only depends on the choices of the polynomials $p(n),q(n)$. Hence, this implies that if $L \in \QMA(2)^{U}$, then it must also be that $L \in \QMA(2)$, since one can use the oracle for all $n < n_0$ making at most $2^{p'(n_0)}$ queries (which is constant for a constant $n_0)$ and simply replace the oracle with the identity operator for all $n \geq n_0$, whilst still having a bounded probability of error (and apply error reduction as in Ref.~\cite{harrow2013testing} to boost it back to $\geq 2/3$). But this contradicts the fact that we have picked $L$ such that it was not in $\QMA(2)$. 
\end{proof}

\orsepadvice*
\begin{proof}
    This follows from a similar proof as~\cref{thm:orsepMP}, but now we choose a binary language $L$ (to be specified later). Let the quantum oracle $U = \{U_x \}_{x \in \{0,1\}^n, n \geq 1}$ with $U_x = e^{2 \pi i \theta_x} \ket{0^n}\bra{0^n} + (\mathbb{I}-\ket{0^n}\bra{0^n})$ be a family of $n$($=|x|$)-qubit unitaries parametrized by some $\theta_x \in [0,1)$, which is given as follows:
\begin{align*}
    \theta_x = \left(1+(-1)^{L(x)} \right) 2^{- p(|x|)-1},
\end{align*}
where $L(x) =1 $ if $x \in L$ and $L(x) = 0$ otherwise. Hence, we have that
\begin{itemize}
    \item If $x \in L$, then $\theta_x = 0$;
    \item If $x \notin L$, then $\theta_x = 2^{-p(|x|)}$.
\end{itemize}
Clearly, $L \in \SBQP^U$ for any choice of language $L$ by the same argument as in~\cref{thm:orsepMP}. Observe that the lower bound of~\cref{clm:QPE} also hold for these specific instances of $U$ where the number of qubits it acts on varies, as the diamond distance only changes with different $n$ because the eigenphases change with different values of $n$ (if $\theta = \epsilon$ for some fixed $\epsilon >0$, then the diamond distance would be fixed for all $n$). Since $\QMA\slash\mathsf{qpoly} \subseteq \PSPACE\slash\mathsf{poly} \neq \ALL$~\cite{aaronson2005qma}, we can use the same argument as in~\cref{thm:orsepMP} to show that there must exist a $L \notin \QMA\slash\mathsf{qpoly}^U$.
\end{proof}
Note that it is not clear if we can combine both unentangled quantum proofs as advice and arrive at a similar oracle separation, as we do not know whether $\QMA(2)\slash \mathsf{qpoly} = \ALL$ or not~\cite{aaronson2005qma}, so it might be that no language $L$ exists such that $L \notin \QMA(2)\slash \mathsf{qpoly}$.

\bibliography{main.bib}

\end{document}

%% file: tikz_figure.tex
\tikzset{every picture/.style={line width=0.75pt}} 

\begin{tikzpicture}[x=0.75pt,y=0.75pt,yscale=-0.8,xscale=0.8]

\draw  [color={rgb, 255:red, 155; green, 155; blue, 155 }  ,draw opacity=1 ] (148,135.75) .. controls (148,78.31) and (194.56,31.75) .. (252,31.75) .. controls (309.44,31.75) and (356,78.31) .. (356,135.75) .. controls (356,193.19) and (309.44,239.75) .. (252,239.75) .. controls (194.56,239.75) and (148,193.19) .. (148,135.75) -- cycle ;
\draw [color={rgb, 255:red, 155; green, 155; blue, 155 }  ,draw opacity=1 ][line width=1.5]    (133,135.75) -- (379,135.75) ;
\draw [shift={(383,135.75)}, rotate = 180] [fill={rgb, 255:red, 155; green, 155; blue, 155 }  ,fill opacity=1 ][line width=0.08]  [draw opacity=0] (6.97,-3.35) -- (0,0) -- (6.97,3.35) -- cycle    ;
\draw [color={rgb, 255:red, 155; green, 155; blue, 155 }  ,draw opacity=1 ][line width=1.5]    (253,256.75) -- (253,16.75) ;
\draw [shift={(253,12.75)}, rotate = 90] [fill={rgb, 255:red, 155; green, 155; blue, 155 }  ,fill opacity=1 ][line width=0.08]  [draw opacity=0] (6.97,-3.35) -- (0,0) -- (6.97,3.35) -- cycle    ;
\draw  [fill={rgb, 255:red, 74; green, 144; blue, 226 }  ,fill opacity=0.27 ] (261,32.25) -- (338.56,79.22) -- (355.44,126.56) -- (334.33,198.83) -- (320,214.3) -- cycle ;
\draw  [draw opacity=0][fill={rgb, 255:red, 74; green, 144; blue, 226 }  ,fill opacity=1 ] (258,32.25) .. controls (258,30.59) and (259.34,29.25) .. (261,29.25) .. controls (262.66,29.25) and (264,30.59) .. (264,32.25) .. controls (264,33.91) and (262.66,35.25) .. (261,35.25) .. controls (259.34,35.25) and (258,33.91) .. (258,32.25) -- cycle ;
\draw  [draw opacity=0][fill={rgb, 255:red, 74; green, 144; blue, 226 }  ,fill opacity=1 ] (335.56,79.22) .. controls (335.56,77.57) and (336.9,76.22) .. (338.56,76.22) .. controls (340.21,76.22) and (341.56,77.57) .. (341.56,79.22) .. controls (341.56,80.88) and (340.21,82.22) .. (338.56,82.22) .. controls (336.9,82.22) and (335.56,80.88) .. (335.56,79.22) -- cycle ;
\draw  [draw opacity=0][fill={rgb, 255:red, 74; green, 144; blue, 226 }  ,fill opacity=1 ] (352.44,126.56) .. controls (352.44,124.9) and (353.79,123.56) .. (355.44,123.56) .. controls (357.1,123.56) and (358.44,124.9) .. (358.44,126.56) .. controls (358.44,128.21) and (357.1,129.56) .. (355.44,129.56) .. controls (353.79,129.56) and (352.44,128.21) .. (352.44,126.56) -- cycle ;
\draw  [draw opacity=0][fill={rgb, 255:red, 74; green, 144; blue, 226 }  ,fill opacity=1 ] (331.33,198.83) .. controls (331.33,197.18) and (332.68,195.83) .. (334.33,195.83) .. controls (335.99,195.83) and (337.33,197.18) .. (337.33,198.83) .. controls (337.33,200.49) and (335.99,201.83) .. (334.33,201.83) .. controls (332.68,201.83) and (331.33,200.49) .. (331.33,198.83) -- cycle ;
\draw  [draw opacity=0][fill={rgb, 255:red, 74; green, 144; blue, 226 }  ,fill opacity=1 ] (317,214.3) .. controls (317,212.64) and (318.34,211.3) .. (320,211.3) .. controls (321.66,211.3) and (323,212.64) .. (323,214.3) .. controls (323,215.96) and (321.66,217.3) .. (320,217.3) .. controls (318.34,217.3) and (317,215.96) .. (317,214.3) -- cycle ;
\draw    (255.81,133.71) -- (286.99,122.14) ;
\draw [shift={(289.8,121.1)}, rotate = 159.65] [fill={rgb, 255:red, 0; green, 0; blue, 0 }  ][line width=0.08]  [draw opacity=0] (3.57,-1.72) -- (0,0) -- (3.57,1.72) -- cycle    ;
\draw [shift={(253,134.75)}, rotate = 339.65] [fill={rgb, 255:red, 0; green, 0; blue, 0 }  ][line width=0.08]  [draw opacity=0] (3.57,-1.72) -- (0,0) -- (3.57,1.72) -- cycle    ;
\draw  [draw opacity=0][dash pattern={on 5.63pt off 4.5pt}][line width=1.5]  (263.2,23.51) .. controls (302.38,27.46) and (338.35,51.84) .. (355.29,90.47) .. controls (375.64,136.9) and (362.22,189.6) .. (325.77,221.09) -- (252,135.75) -- cycle ; \draw  [color={rgb, 255:red, 208; green, 2; blue, 27 }  ,draw opacity=1 ][dash pattern={on 5.63pt off 4.5pt}][line width=1.5]  (263.2,23.51) .. controls (302.38,27.46) and (338.35,51.84) .. (355.29,90.47) .. controls (375.64,136.9) and (362.22,189.6) .. (325.77,221.09) ;  
\draw [color={rgb, 255:red, 208; green, 2; blue, 27 }  ,draw opacity=1 ][fill={rgb, 255:red, 208; green, 2; blue, 27 }  ,fill opacity=1 ][line width=1.5]    (323,217.55) -- (328.27,223) ;
\draw [color={rgb, 255:red, 208; green, 2; blue, 27 }  ,draw opacity=1 ][fill={rgb, 255:red, 208; green, 2; blue, 27 }  ,fill opacity=1 ][line width=1.5]    (263,20.09) -- (261.73,27.18) ;

\draw (362.5,167.5) node [anchor=north west][inner sep=0.75pt]    {$\textcolor[rgb]{0.82,0.01,0.11}{\Theta ( U)}$};
\draw (259,104.5) node [anchor=north west][inner sep=0.75pt]    {$D$};

\end{tikzpicture}